\pdfoutput=1
\documentclass[a4paper,abstract=true]{scrartcl}
\usepackage[a4paper,margin=2.5cm,includefoot]{geometry}

\usepackage[utf8]{inputenc}
\usepackage[english]{babel}
\usepackage[T1]{fontenc}

\usepackage{lmodern}
\usepackage{amsmath,amssymb,mathtools,bm}
\usepackage{csquotes}
\usepackage{amstext}
\usepackage{amsthm}
\usepackage{bbm}
\usepackage{enumerate}
\usepackage{hyperref}
\usepackage[nameinlink,capitalise]{cleveref}
\usepackage{dsfont}
\usepackage{comment}
\usepackage{color}
\usepackage{pgfplots}\pgfplotsset{compat=1.8}

\usepackage{fancyvrb}

\makeatletter

\numberwithin{equation}{section}
\newtheorem{thm}{Theorem}[section]
\newtheorem{lem}[thm]{Lemma}
\newtheorem{prop}[thm]{Proposition}
\newtheorem{cor}[thm]{Corollary}

\theoremstyle{definition}
\newtheorem{defn}[thm]{Definition}

\renewcommand*{\thehyp}{\Alph{hyp}}

\theoremstyle{remark}
\newtheorem{rem}[thm]{Remark}
\newtheorem{ex}[thm]{Example}

\crefname{hyp}{Hypothesis}{Hypotheses}
\Crefname{hyp}{Hypothesis}{Hypotheses}
\crefname{lem}{Lemma}{Lemmas}
\Crefname{lem}{Lemma}{Lemmas}
\crefname{thm}{Theorem}{Theorems}
\Crefname{thm}{Theorem}{Theorems}
\crefname{prop}{Proposition}{Propositions}
\Crefname{prop}{Proposition}{Propositions}
\crefname{cor}{Corollary}{Corollaries}
\Crefname{cor}{Corollary}{Corollaries}
\crefname{enumi}{}{}
\Crefname{enumi}{}{}
\creflabelformat{enumi}{#2(#1)#3}
\crefname{equation}{}{}
\Crefname{equation}{}{}
\crefname{rem}{Remark}{Remarks}
\Crefname{rem}{Remark}{Remarks}

\makeatletter
\renewcommand{\@upn}{} %
\makeatother

\usepackage[inline]{enumitem}
\newlist{enumthm}{enumerate}{1} %
\setlist[enumthm]{label=\upshape(\roman*),ref=\thethm~(\roman*)}  %
\crefalias{enumthmi}{thm} %
\newlist{enumcor}{enumerate}{1}
\setlist[enumcor]{label=\upshape(\roman*),ref=\thecor~(\roman*)}
\crefalias{enumcori}{cor}
\newlist{enumlem}{enumerate}{1}
\setlist[enumlem]{label=\upshape(\roman*),ref=\thelem~(\roman*)}
\crefalias{enumlemi}{lem}
\newlist{enumprop}{enumerate}{1}
\setlist[enumprop]{label=\upshape(\roman*),ref=\theprop~(\roman*)}
\crefalias{enumpropi}{prop}
\newlist{enumhyp}{enumerate}{1}
\setlist[enumhyp]{label=\upshape(\roman*),ref=\thehyp~(\roman*)}
\crefalias{enumhypi}{hyp}
\newlist{enumproof}{enumerate*}{1}
\setlist[enumproof]{label=\upshape(\roman*)}
\newlist{enumdef}{enumerate}{1}
\setlist[enumdef]{label=\upshape(\roman*),ref=\thedefn~(\roman*)}
\crefalias{enumdefi}{defn}
\newlist{enumrem}{enumerate}{1}
\setlist[enumrem]{label=\itshape(\alph*),ref=\therem~(\alph*)}
\crefalias{enumrem}{rem}

\makeatletter
\newcounter{subcreftmpcnt} %
\newcommand\romansubformat[1]{(\roman{#1})} %
\def\subcref{\@ifstar\@@subcref\@subcref}
\newcommand\@subcref[2][\romansubformat]{%
	\ifcsname r@#2@cref\endcsname
	\cref@getcounter {#2}{\mylabel}%
	\setcounter{subcreftmpcnt}{\mylabel}%
	\hyperref[#2]{\romansubformat{subcreftmpcnt}}%
	\else ?? \fi}   
\newcommand\@@subcref[2][\romansubformat]{%
	\ifcsname r@#2@cref\endcsname
	\cref@getcounter {#2}{\mylabel}%
	\setcounter{subcreftmpcnt}{\mylabel}%
	\romansubformat{subcreftmpcnt}%
	\else ?? \fi}   
\makeatother

\makeatletter
\DeclareRobustCommand{\crefnosort}[1]{%
	\begingroup\@cref@sortfalse\cref{#1}\endgroup
}
\makeatother

\def\endstepsymbol{$\lozenge$}
\def\endclaimsymbol{$\lozenge$}
\newcounter{proofstep}
\AtBeginEnvironment{proof}{\setcounter{proofstep}{0}}

\crefname{proofstep}{Step}{Steps}
\Crefname{proofstep}{Step}{Steps}
\newcounter{proofclaim}
\AtBeginEnvironment{proof}{\setcounter{proofclaim}{0}}

\crefname{proofclaim}{Claim}{Claims}
\Crefname{proofclaim}{Claim}{Claims}

\newcommand{\cF}{{\mathcal F}}

\newcommand{\cK}{{\mathcal K}}

\newcommand{\fh}{{\mathfrak h}}

\newcommand{\BC}{{\mathbb C}}

\newcommand{\BN}{{\mathbb N}}
\newcommand{\BR}{{\mathbb R}}

\newcommand{\BZ}{{\mathbb Z}}

\usepackage{dsfont} %

\newcommand{\dsone}{{\mathds 1}}

\usepackage{mathrsfs}

\newcommand{\sfc}{{\mathsf c}}
\newcommand{\sfd}{{\mathsf d}}
\newcommand{\sfi}{{\mathsf i}}

\newcommand{\IN}{\BN}\newcommand{\IZ}{\BZ}\newcommand{\IR}{\BR}\newcommand{\IC}{\BC}

\newcommand{\hs}{\fh}

\newcommand{\ph}{\varphi}

\renewcommand{\i}{\sfi}\newcommand{\Id}{\dsone} \renewcommand{\d}{\sfd}

\renewcommand{\Re}{\operatorname{Re}}\renewcommand{\Im}{\operatorname{Im}}

\newcommand{\supp}{\operatorname{supp}}

\DeclareFontFamily{U}{mathx}{\hyphenchar\font45}
\DeclareFontShape{U}{mathx}{m}{n}{
	<5> <6> <7> <8> <9> <10>
	<10.95> <12> <14.4> <17.28> <20.74> <24.88>
	mathx10
}{}
\DeclareSymbolFont{mathx}{U}{mathx}{m}{n}
\DeclareFontSubstitution{U}{mathx}{m}{n}
\DeclareMathAccent{\widecheck}{0}{mathx}{"71}
\DeclareMathAccent{\wideparen}{0}{mathx}{"75}

\DeclareFontFamily{OMX}{MnSymbolE}{}
\DeclareFontShape{OMX}{MnSymbolE}{m}{n}{
	<-6>  MnSymbolE5
	<6-7>  MnSymbolE6
	<7-8>  MnSymbolE7
	<8-9>  MnSymbolE8
	<9-10> MnSymbolE9
	<10-12> MnSymbolE10
	<12->   MnSymbolE12}{}
\DeclareSymbolFont{mnlargesymbols}{OMX}{MnSymbolE}{m}{n}
\SetSymbolFont{mnlargesymbols}{bold}{OMX}{MnSymbolE}{b}{n}
\DeclareMathDelimiter{\llangle}{\mathopen}{mnlargesymbols}{'164}{mnlargesymbols}{'164}
\DeclareMathDelimiter{\rrangle}{\mathclose}{mnlargesymbols}{'171}{mnlargesymbols}{'171}
\DeclareMathDelimiter{\lsem}{\mathopen}{mnlargesymbols}{'102}{mnlargesymbols}{'102}
\DeclareMathDelimiter{\rsem}{\mathclose}{mnlargesymbols}{'107}{mnlargesymbols}{'107}
\DeclareMathDelimiter{\langlebar}{\mathopen}{mnlargesymbols}{'152}{mnlargesymbols}{'152}
\DeclareMathDelimiter{\ranglebar}{\mathclose}{mnlargesymbols}{'157}{mnlargesymbols}{'157}
\DeclareMathDelimiter{\lWavy}{\mathopen}{mnlargesymbols}{'137}{mnlargesymbols}{'137}
\DeclareMathDelimiter{\rWavy}{\mathopen}{mnlargesymbols}{'137}{mnlargesymbols}{'137}

\newcommand{\chr}{\mathbf 1}
\newcommand{\abs}[1]{\lvert#1\lvert}
\newcommand{\Norm}[1]{\left\lVert#1\right\lVert}

\newcommand{\FS}{\cF}

\newcommand{\nn}[1]{\Norm{#1}}
\newcommand{\nnw}[1]{ \nn{#1}_{(w_n)}}

\renewcommand{\:}{\colon}

\usepackage{xcolor}\usepackage{cancel}
\definecolor{green}{rgb}{0.0, 0.5, 0.5}
\definecolor{yellow}{rgb}{0.5, 0.5, 0}
\definecolor{lgray}{gray}{0.9}
\definecolor{llgray}{gray}{0.95}
\definecolor{lllgray}{gray}{0.975}

\definecolor{darkcerulean}{rgb}{0.03, 0.27, 0.49} 	\definecolor{darkcoral}{rgb}{0.8, 0.36, 0.27}
\usepgfplotslibrary{colormaps}
\pgfplotsset{/pgfplots/colormap/bluewhitered/.style={/pgfplots/colormap={bluewhitered}{color(0cm)=(darkcerulean); color(0.5cm)=(white); color(1cm)=(darkcoral)}}}
\pgfplotsset{/pgfplots/colormap/bluered/.style={/pgfplots/colormap={bluered}{color(0cm)=(darkcerulean); color(1cm)=(darkcoral)}}}
\definecolor{darkgreen}{rgb}{0,.39,0}
\usepackage[outline]{contour}

\usepgfplotslibrary{external}
\tikzset{external/system call={pdflatex \tikzexternalcheckshellescape -extra-mem-top=10000000 -extra-mem-bot=10000000 -halt-on-error -interaction=batchmode -jobname "\image" "\texsource"}}

\newcommand{\Af}{\mathfrak{A}}
\newcommand{\tr}{{\operatorname{tr}}}

\newcommand{\Cci}{C_\sfc^\infty}
\newcommand{\ketbra}[2]{ {|{#1}\rangle}\hspace*{-0.9mm}{\langle{#2}|}  }
\newcommand{\thmenum}{\leavevmode\vspace{-\baselineskip+4mm} \vspace*{-\medskipamount}}

\usepackage[backend=biber,style=alphabetic,giveninits=true,maxnames=5,maxalphanames=4, ,url=true,isbn=false,eprint=false,alldates=year]{biblatex}
\DeclareFieldFormat{doilink}{%
	\iffieldundef{doi}{#1}%
	{\href{http://dx.doi.org/\thefield{doi}}{#1}}}

\DeclareBibliographyDriver{article}{%
	\usebibmacro{bibindex}%
	\usebibmacro{begentry}%
	\usebibmacro{author/translator+others}%
	\setunit{\labelnamepunct}\newblock
	\usebibmacro{title}%
	\newunit
	\printlist{language}%
	\newunit\newblock
	\usebibmacro{byauthor}%
	\newunit\newblock
	\usebibmacro{bytranslator+others}%
	\newunit\newblock
	\printfield{version}%
	\newunit\newblock
	\usebibmacro{in:}%
	\printtext[doilink]{%
		\usebibmacro{journal+issuetitle}%
		\newunit
		\usebibmacro{byeditor+others}%
		\newunit
		\usebibmacro{note+pages}%
	}%
	\newunit\newblock
	\iftoggle{bbx:isbn}
	{\printfield{issn}}
	{}%
	\newunit\newblock
	\usebibmacro{doi+eprint+url}%
	\newunit\newblock
	\usebibmacro{addendum+pubstate}%
	\setunit{\bibpagerefpunct}\newblock
	\usebibmacro{pageref}%
	\usebibmacro{finentry}}

\newbibmacro{string+doi}[1]{%
	\iffieldundef{doi}{#1}{\href{https://doi.org/\thefield{doi}}{#1}}}
\DeclareFieldFormat{title}{\usebibmacro{string+doi}{\mkbibemph{#1}}}
\DeclareFieldFormat[inproceedings]{title}{\usebibmacro{string+doi}{\mkbibquote{#1}}}
\DeclareFieldFormat[incollection]{title}{\usebibmacro{string+doi}{\mkbibquote{#1}}}
\DeclareFieldFormat[unpublished]{title}{\usebibmacro{string+doi}{\mkbibquote{#1}}}

\renewbibmacro{in:}{%
	\ifentrytype{article}{}{\printtext{\bibstring{in}\intitlepunct}}}

\newcommand{\Lb}{\mathcal{L}}

\newcommand{\kh}{\mathfrak{K}}
\newcommand{\lin}{\operatorname{lin}}
\renewcommand{\H}{\mathcal{H}}

\newcommand{\cs}[1]{\left< #1 \right>}
\newcommand{\sgn}{\operatorname{sgn}}
\newcommand{\Def}{\mathcal{D}}
\newcommand{\Afe}{\widehat{\Af}}
\newcommand{\gw}{\widetilde g}
\newcommand{\fw}{\widetilde f}
\newcommand{\Pa}[1]{P^-_{#1}}
\newcommand{\Vb}{\mathbf{V}}
\newcommand{\PN}{\Pi}
\newcommand{\An}{A_n}
\newcommand{\khs}{\kh^\otimes}
\newcommand{\FSo}[1]{\FS_{#1}^\otimes}
\newcommand{\ssum}[1]{\sum_{\substack{#1}}}
\newcommand{\B}{\mathcal{B}}

\begin{document}

\title{On the thermodynamic limit of interacting fermions in the continuum}
\author{Oliver Siebert\footnote{Current affiliation: UC Davis, Department of Mathematics, Davis, CA 95616-8633, USA,  \href{mailto:osiebert@ucdavis.edu}{osiebert@ucdavis.edu}} \\ \small{Department of Mathematics, University of Tübingen} \\ \small{Auf der Morgenstelle 10, 72076 Tübingen, Germany}}

\maketitle

\begin{abstract}
		We study the dynamics of non-relativistic fermions in $\IR^d$ interacting through a pair potential.  Employing methods developed by Buchholz in the framework of resolvent algebras, we consider an extension of the CAR algebra where the dynamics acts as a group of $*$-automorphisms, which are continuous in time in all sectors for fixed particle numbers. In addition, we construct a $C^*$-dynamical system by identifying a suitable dense subalgebra. Finally, we
		briefly discuss how this framework could be used to construct KMS states in the future.
\end{abstract}

\section{Introduction}

The existence of the thermodynamic limit with a corresponding continuous Heisenberg time evolution and a set of equilibrium (KMS) states  is of fundamental importance for the description of large, infinitely extended many-body systems \cite{sakai1991operator}. For example, in algebraic quantum statistical mechanics, one studies the dynamical behavior of open quantum systems consisting of a small subsystem interacting with  an infinitely extended heat bath at thermal equilibrium  \cite{attal2006open,jakpillet3}. The latter are usually modeled as bosonic or fermionic \cite{jakvsic2002non} ideal gases in $\IR^3$, for which the KMS states and the corresponding GNS representations (Araki-Woods and Araki-Wyss \cite{araki_wyss}) are explicitly known.

In the bosonic case, even the free dynamics is not strongly continuous on the Weyl algebra, so one has to resort to certain subalgebras or representation-dependent $W^*$-dynamical systems. In contrast, for fermions, the free dynamics is strongly continuous on the entire CAR algebra, which leads to the convenient setting of $C^*$-dynamical systems. However, this seems to be a peculiarity of the free case, and it is generally believed that including interactions between particles violates this property. For example, one could think of states that develop arbitrarily large local densities or energies within a finite time by filling the phase space with an increasing number of fermions at the same location with different momenta. Due to the non-zero repulsion or attraction at fixed distances, this would lead to arbitrarily strong attractions on individual fermions and thus to a discontinuity in the time evolution, cf. \cite[p. 354]{BR2}. 
 
By imposing ultraviolet cutoffs in the interaction, it can be shown that strong continuity still holds for a certain class of potentials. This was first addressed by Narnhofer and Thirring \cite{narnhofer1990quantum,narnhofer1991galilei} with a UV regularization that preserves Galilean invariance. They also demonstrated the existence of KMS states, which were later proven to be unique in certain regimes \cite{jakel1995uniqueness}.
A similar result, using a regularization with smeared-out interactions between the particles, was proved in \cite{GebertNachtergaeleReschkeSims.2020,hinrichs2024lieb} based on Lieb-Robinson bounds \cite{LiebRobinson.1972}.
Continuous fermionic systems with magnetic fields, where a UV regularization naturally appears via localized frames, was recently considered in \cite{bachmann2024liebrobinson}, based on a new Lieb-Robinson bound.
 In fact, the existence of the thermodynamic limit is known to be one of the primary applications of Lieb-Robinson bounds; see \cite{td1,td2,td3,td4,irreversible} for applications in lattice models. While Lieb-Robinson bounds have been successfully established for lattice fermions \cite{gluza2016equilibration,NachtergaeleSimsYoung.2018}, the situation in the continuum is much more challenging due to local UV divergences. Specifically, the current works \cite{GebertNachtergaeleReschkeSims.2020,hinrichs2024lieb} require bounded perturbations when restricting to a finite box. Since this condition does not hold when the UV cutoff is removed, and since the Hamiltonian only converges in the strong resolvent sense, these methods fail in this context.

In \cite{bh1,bh2} Buchholz proposed an alternative approach for constructing a $C^*$-dynamical system for continuous interacting bosons.  It is based on his and Grundling's resolvent algebra approach \cite{bh}, where the algebra is generated by resolvents of the field. Therefore, in singular states with infinite particle densities the observables simply vanish, circumventing one of the main problems for bosons. The pure many-body problem of an unbounded interaction term, growing with the number of particles, still remains, and was treated in  \cite{bh1} by considering the problem in all $n$-particle sectors. He used the explicit structure  of the particle number preserving observables of the resolvent algebra. The $n$-particle sectors are generated by a mixture of compact operators and identities on the different tensor factors, and different sectors are connected with certain coherent relations. With a slight extension, one obtains an algebra which is invariant under the dynamics, and where the dynamics acts pointwise continuously in each particle sector. Restricting to a subalgebra generated by time averages then leads to an even pointwise norm continuous dynamics (with respect to the operator norm on the full Hilbert space, as compared to restrictions to $n$ particles), and so to a $C^*$-dynamical system. In \cite{bh2} this technique was generalized to observables which are not particle number preserving.  

In contrast to  other constructions of the thermodynamic limit by means of localized systems (cf. e.g. \cite{haag1967equilibrium,BR2,sakai1991operator}) this approach allows for the direct study of the infinite-volume system from the outset. This eliminates the need for approximations with increasingly large finite boxes with sharp boundary conditions, and varying volume-dependent algebras. Instead, one can work with arbitrary smooth trapping potentials for the localization, cf. also \cite{buchholzyngvason2024} for recent results in this context for free bosons. 
In particular, when trying to construct equilibrium states, one can work with a sequence of regularized models, where all the corresponding dynamics and states, as well as their limits, are defined within the same algebra. 

The purpose of the present work is to apply these methods in the fermionic case, as already briefly suggested in \cite{bh1}. Using well-established structural results about the CAR algebra \cite{bratteli1972inductive}, one can easily show that the $n$-particle sectors and the coherence relations in the standard CAR algebra have an analogous structure as in the resolvent algebra. Thus, the previous techniques can be used to tackle the fermionic many-body problem directly in a slightly extended CAR algebra, which is invariant under the dynamics. As in the bosonic case, we prove that the dynamics is continuous in the $n$-particle sectors and we show that the dense (with respect to the $n$-particle norms) subalgebra generated by time averages gives rise to a $C^*$-dynamical system. Our analysis also includes the case of observables which are not particle number preserving, which, not surprisingly, requires much shorter proofs than in the bosonic resolvent case \cite{bh2}. 

This approach addresses the problem of finding a sufficiently small extension of the CAR algebra which is still invariant under the dynamics and, additionally, an algebra as large as possible where the interacting dynamics is still pointwise norm continuous. This algebra is generated by time averages and its size can be determined by the topology in which it is dense in the CAR algebra. In our setting this topology is the convergence in all $n$-particle sectors. This is in contrast to the trivial approach of considering the (weaker) strong operator topology (see also \Cref{rem:sot}), in which the dynamics is a priori continuous due to the existence of a self-adjoint generator. In this case, one can extend the CAR algebra to all bounded operators and consider time averages as well, leading to a $C^*$-dynamical system which is dense in the CAR algebra with respect to the strong operator topology. Note in particular that the extension of the CAR algebra is necessary for invariance under the interacting dynamics. It does not lead to the omission of the problematic states responsible for the discontinuity of the dynamics. Instead, this problem is solved by averaging over time. 

In addition, we show that the current scheme provides a potential framework for the construction of KMS states in the strict sense of $C^*$-dynamical systems. As in the free case, one can expect to obtain them by limit or accumulation points of Gibbs states, when a suitable localization tends to zero. To this end, we briefly discuss sequences of vanishing trapping potentials, where it can be shown that in the non-interacting case the corresponding Gibbs states converge to the well-known quasi-free KMS state, cf. \cite{bh1}. In the interacting case, the challenge is to verify the KMS condition of the limit points, since we have no norm convergence of the dynamics. However, one can show convergence in all $n$ particle sectors or in stronger topologies with weighted supremum norms depending on the number of particles. Therefore, the chosen setup is naturally promising for the study of KMS states.

In the subsequent section we present the formalism, the model and the main theorems, followed by a short discussion about KMS states. In \Cref{sec:structure} we study the structure of the $n$-particle sectors of the particle number preserving CAR algebra and discuss its corresponding extension. Then in \Cref{sec:invariance} we discuss the dynamics and its continuity on this algebra. There we also prove the first main theorem for the particle number preserving case. The general case, i.e., the counterpart of \cite{bh2}, is then treated in \Cref{sec:full}.

\section{Model and main results}

Before introducing the Hamiltonian and stating the main results, we briefly recall some standard aspects of the formalism of fermionic Fock spaces. Readers not familiar with the material may also consult  \cite{BR2,Arai.2018} as thorough references. In the following, let $\Lb(\H)$ and $\cK(\H)$ denote the bounded and compact operators on a Hilbert space $\H$, and let $\hs := L^2(\IR^d)$ be the Hilbert space describing one particle or fermion. The antisymmetric orthogonal projection  $\Pa{n}$  in $\hs^{\otimes n}$ is given by
\begin{align*}
	f_1 \wedge \ldots \wedge f_n := \Pa{n} (f_1 \otimes \cdots \otimes f_n) := \frac{1}{n!} \sum_{\sigma \in S_n} (-1)^\sigma f_{\sigma(1)} \otimes \cdots \otimes f_{\sigma(n)}, \qquad f_j \in \hs.
\end{align*}
 The fermionic Fock space is defined as
\begin{align*}
	\cF =  \bigoplus_{n=0}^\infty \cF_n, \qquad \cF_n = \begin{cases}
		\IC &: n= 0, \\_{}
		 \Pa{n} \hs^{\otimes n} &: n > 0.
	\end{cases} 
\end{align*}
We can write each element $\psi \in \cF$ as a sequence $(\psi_n)_{n \in \IN_0}$, where $\psi_0 \in \IC$ and $\psi_n \in L^2_a((\IR^d)^{\times n})$, $n \in \IN$, the antisymmetric $L^2$ functions in $n$ variables. The vacuum vector is given by $\Omega := (1,0,0,\ldots)$. We will also use the notation  $\FS_{\leq n} := \bigoplus_{k=0}^n \FS_k$. 

For $f \in \hs$, the creation and annihilation operators $a^*(f), a(f) \in \Lb(\cF)$ are defined as
\begin{align}
	(a^*(f) \psi)_n(x_1, \ldots, x_n) &:= \begin{cases}
		0 &: n= 0, \\
		n^{-1/2} \sum_{i=1}^n (-1)^{i+1} f(x_i)  \psi_{n-1}(x_1, \ldots, \widehat{x_i}, \ldots  ,x_n)  &: n \in \IN, 
	\end{cases} \label{eq:creation op} \\
	(a(f) \psi)_n(x_1, \ldots, x_n) &:= \sqrt{n+1} \int \overline{f(x)} \psi_{n+1}(x,x_1, \ldots, x_n)  \d x, \qquad n \in \IN_0, \label{eq:annihilation op} 
\end{align}
where $\psi \in \cF$. The operators are adjoint to each other and satisfy the well-known canonical anticommutation relations (CAR),
\begin{align}
	\label{eq:car}
	\{a(f),a(g)\} = 0, \qquad \{a(f), a^*(g)\} = \cs{f,g}, \qquad f,g \in \hs. 
\end{align}
In particular, \eqref{eq:car} implies that $\nn{a^*(f)} = \nn{a(f)} = \nn f$. We will also write $a^\#$ as a symbol which can resemble either $a$ or $a^*$. The local number operators are  given by $n(f) := a^*(f) a(f)$.

The CAR algebra with respect to $\hs$ is defined as
\begin{align*}
	\Af := C^*(a(f) : f \in \hs ) \subseteq \Lb(\cF),
\end{align*}
the $C^*$-algebra generated by $a(f),~f \in \hs$, i.e., the norm closure of the polynomial algebra in $a(f)$ and $a^*(g)$, $f,g \in \hs$. Furthermore, we introduce the particle number preserving CAR algebra $\Af_0 \subset \Af$ as the unital $C^*$-subalgebra of $\Af$, which preserves all particle sectors $\FS_n$, $n \in \IN_0$. It can also be written as
\begin{align}
\label{eq:pp algebra}
\Af_0 =  \overline { \operatorname{lin} \{ a^*(f_1) \cdots a^*(f_n) a(g_1) \cdots a(g_n) : n\in \IN_0,~ f_1, \ldots, f_n, g_1, \ldots, g_n \in \hs \}  }.
\end{align}

We want to study the dynamics on $\Af$ and $\Af_0$ for an interacting fermionic system, given by the following Hamiltonian. Let $V \: \IR^d \rightarrow \IR$ be a real, continuous and even function, which vanishes at infinity. The Hamiltonian on the $n$-particle sector is a self-adjoint operator with domain  $\Def(H_n) := \Pa{n} (H^2(\IR^d)^{\otimes n}) \subseteq \FS_n$  given by
\begin{align}
		\label{eq:Hamiltonian n}
		H_n := \sum_{i=1}^n (-\Delta_i) +  \sum_{\substack{i,j=1,\\i\not= j}}^n V_{ij} ,
\end{align}
where $-\Delta_i$ denotes the operation of the Laplacian $-\Delta$ on the $i$-th tensor factor and
\begin{align*}
 (V_{ij} \psi)(x_1, \ldots, x_n) := V(x_i-x_j)\psi(x_1, \ldots,  x_n). 
\end{align*}
 Note that each operator $V_{ij}$ is bounded and we have a bounded perturbation for fixed particle number $n$. However, in the whole Fock space $\FS$, the perturbation is unbounded and grows like $n(n-1)$. 

The complete Hamiltonian on Fock space $\cF$ is then given by 
\begin{align}
	\label{eq:Hamiltonian}
	H := \bigoplus_{n=1}^\infty H_n
\end{align}
and well-known to be self-adjoint on its natural domain \cite{GebertNachtergaeleReschkeSims.2020}, i.e.,
\begin{align*}
	\Def(H) := \{ \psi \in \cF: \psi_n \in \Def(H_n),~ \sum_{n=1}^{\infty} \nn{ H_n \psi_n}^2 < \infty   \}.
\end{align*}

We consider the dynamics with respect to $H$, $\alpha_t(A) = e^{\i t H} A e^{-\i t H}$, $A\in \Lb(\FS)$, on an extended version  of $\Af$. To this end, let us introduce a family of seminorms $\nn{\cdot}_n$, $n \in \IN_0$, on $\Lb(\cF)$ given by
\begin{align*}
	\nn{A}_n := \nn{A|_{\cF_n}} = \nn{A \PN_n}, \qquad A \in \Lb(\cF),
\end{align*}
where $\PN_n$ denotes the orthogonal projection onto the subspace of $n$ particles. The algebras $\Af$ and $\Af_0$ equipped with the seminorms $\nn{\cdot}_n$ form a locally convex space. Note that this topology is (strictly) weaker than the norm topology but stronger than the strong operator topology. For instance, for any orthonormal basis $(f_k)_{k\in\IN}$ of $\hs$,  consider the sequence defined by $A_k := n(f_1) \cdots n(f_k) \in \Af_0$. Then $A_k \overset{k\to\infty}\to 0$ in the seminorms topology but not in the norm topology. 

Let $\Afe_0 \subseteq \Lb(\cF)$ be the closure of $\Af_0$ with respect to the $\nn\cdot_n$ seminorms. This is a $C^*$-algebra as well, since it is closed in the seminorms and therefore also closed in the norm topology. It should be noted that the extension does not coincide with the algebra itself, see \cref{ex:counter}. 

For the full observable algebra we can proceed as follows. For each $k \in \IN$, let
\begin{align*}
	\Af_k &:= \lin \{ B a^*(f_1) \cdots a^*(f_k): B \in \Af_0,~f_1,\ldots,f_k \in \hs   \}, \\
	\Af_{-k} &:= \lin \{ B a(f_1) \cdots a(f_k): B \in \Af_0,~f_1,\ldots,f_k \in \hs   \},
\end{align*}
i.e., the algebras, whose elements add and remove $k$ particles, respectively. 
Let $\Afe_k \subseteq \Lb(\FS)$, $k \in \IZ$, be the closures of $\Af_k$ with respect to the $\nn \cdot_n$ seminorms. Then we define
\begin{align}
	\label{eq:Afe characterization}
	\Afe :=  \overline{ \bigoplus_{k \in \IZ} \Afe_k },
\end{align} 
where the closure is with respect to the operator norm. We see  in \Cref{sec:full} that  $\Afe$ is in fact a $C^*$-algebra.  Note that  $\Af \subseteq \Afe$, $\Afe_0 \subseteq \Afe$, and since the seminorms topology is weaker than the operator norm topology, $\Afe_k$ is also closed in the operator norm but $\Afe$ is not closed in the seminorms.

 Now we have all the ingredients to formulate our main result. 
\begin{thm}
	\label{th:main}
	The dynamics $(\alpha_t)_{t\in\IR}$ leaves $\Afe$ (resp. $\Afe_0$) invariant and therefore forms a group of *-automorphisms on $\Afe$ (resp. $\Afe_0$). Furthermore, for each $A \in \Afe$ (resp. $\Afe_0$), the map $\IR \longrightarrow \Afe$ (resp. $\IR \longrightarrow \Afe_0$), $t \mapsto \alpha_t(A)$, is continuous with respect to the topology induced by the $\nn \cdot_n$, $n \in \IN$, seminorms.
\end{thm}
The proof will be given in  \Cref{sec:invariance} for the case $\Afe_0$ and in \cref{sec:full} for the case $\Afe$. Using time averages as in \cite{bh1,bh2}, we obtain a $C^*$-algebra, which is dense in $\Afe$ in the seminorms topology and where the time evolution is continuous in norm. 
\begin{cor}
	\label{th:main coro}
	There exists a unital subalgebra $\Afe_H \subseteq \Afe $ (resp. $\Afe_{0,H} \subseteq \Afe_0$) depending on $H$ which is dense in $\Afe$ (resp. $\Afe_0$)  with respect to the $\nn \cdot_n$, $n \in \IN$, seminorms topology and invariant under $(\alpha_t)_{t\in\IR}$. The map $\IR \longrightarrow \Afe_H$ (resp. $\IR \longrightarrow  \Afe_{0,H}$), $t \mapsto \alpha_t(A)$, is continuous in norm for all $A \in \Afe_H$ (resp. $A \in \Afe_{0,H}$). Therefore, $(\Afe_H, (\alpha_t)_{t\in\IR})$ (resp. $(\Afe_{0,H},  (\alpha_t)_{t\in\IR})$) forms a $C^*$-dynamical system. 
\end{cor}
\begin{rem}
	\label{rem:sot}
	Considering time averages is a common technique to obtain continuity in norm, see e.g. \cite{thermalionization} for an application in the free bosonic case. In our setting, it could also be applied directly to $\Lb(\H)$ with the strong operator topology without having to prove \cref{th:main}. However, the resulting subalgebra would only be dense in the weaker strong operator topology.
\end{rem}
\subsubsection*{Existence of KMS states}
The significance of our result lies in the fact that it provides a potential framework to construct KMS states in the sense of $C^*$-dynamical systems. Candidates for such states can be obtained as weak-$*$-limit points of states of suitably regularized models. One possibility already suggested in \cite{bh1} is to add trapping potentials which guarantee the existence of Gibbs states.
 Then one considers sequences of states where these potentials converge to zero. 

More precisely, we introduce a modified form of the Hamiltonian \eqref{eq:Hamiltonian n} by replacing the Laplacian by a harmonic oscillator. So we define for $L> 0$, 
\begin{align*}
	H^L := \bigoplus_{n=1}^\infty H^L_n, \qquad H^L_n := \sum_{i=1}^n \left(-\Delta_i + \frac{x_i^2}{L^4} \right) + \sum_{\substack{i,j=1,\\i\not= j}}^n V_{ij}.
\end{align*}
This operator has discrete spectrum because of the Golden-Thompson inequality. Thus, the corresponding Gibbs states 
\begin{align*}
	\omega^\beta_L(A) := \frac{\tr(A e^{-\beta H^L})}{\tr( e^{-\beta H^L})}, \qquad A \in \Afe,
\end{align*}
exist. In the free case $V=0$ it is expected these states converge to the well known quasi-free states, similarly as in the bosonic setting \cite[p. 974]{bh1}. In the interacting case $V \not= 0$, there are a priori only weak-$*$-limit points by Banach-Alaoglu. Hence, we find sequences $(L_k)_{k \in \IN}$ such that $L_k \to \infty$ and
\begin{align*}
	\omega := \lim_{k\to\infty} \omega_{L_k}
\end{align*}
exists in the weak-$*$-topology in $\Afe^*$. 
 
The verification of the KMS condition for $\omega$ is in fact more challenging. Let  $\alpha^L_t$ be the time evolution with respect to the Hamiltonian $H^L$. It is relatively easy to see that a sufficient condition is 
\begin{align}
	\label{eq:sufficient}
\lim_{k\to\infty} \omega_{L_k}(A\alpha^{L_k}_t(B)) =  \omega(A\alpha_t(B)), \qquad A,B \in \Afe.
\end{align}
\begin{prop}
	\label{th:KMS property}
	Assume that \eqref{eq:sufficient} holds. 
 Then we have, for all functions $f$ with $\widehat f \in \Cci(\IR)$ and all $A,B \in \Afe$,
	\begin{align}
		\label{eq:KMS condition}
		\int_{\IR} f(t) \omega(A \alpha_t(B)) \d t = \int_{\IR} f(t + \i \beta) \omega(\alpha_t(B) A) \d t. 	
	\end{align}
	In particular, $\omega$ is a $\beta$-KMS state for the time evolution  $(\alpha_t)_{t\in \IR}$ if we restrict it to $\Afe_{0,H}$. 
\end{prop}
\begin{proof}
	We use  \cite[Prop. 5.3.12]{BR2} for the characterization of KMS states.	By \cite[Proposition 2.5.22]{BR1} there exists a dense algebra $\Afe^L \subset \Afe$ of analytic elements, i.e.,  for all $B \in \Afe^L$, there exists an entire analytic (with respect to the seminorms topology) extension $\IC \ni z \mapsto  \alpha^L_z(B)$. In particular, for all $A, B \in \Afe^L$, $\IC \rightarrow \IC$, $z \mapsto \omega_L(A \alpha^L_z(B) )$ is entire analytic. As $e^{-\beta H^L}$ is trace-class, a standard calculation shows $\omega_L(A \alpha^L_t(B)) = \omega_L(\alpha^L_{t-\i \beta}(B) A)$ for all $t \in \IR$. Then for any $f$ with $\widehat f \in \Cci(\IR)$ we obtain by Paley-Wiener  \cite[Prop. 5.3.11]{BR2} that $f$ is entire analytic, hence also $z \mapsto f(z) \omega_L(\alpha^L_z(B) A)$. Then, again by Paley-Wiener, the latter function restricted to $\{z \in \IC : \abs{\Im z} \leq \beta \}$  decays faster than $\abs{\Re z}^{-2}$ as $\abs{\Re z} \to \infty$. Therefore, Cauchy's integral theorem yields
	\begin{align}
		\label{eq:integral kms}
		\int_{\IR} f(t) \omega_L(A\alpha^L_t(B)) \d t = \int_{\IR} f(t) \omega_L(\alpha^L_{t-\i \beta}(B) A) \d t = \int_{\IR} f(t + \i \beta) \omega_L(\alpha^L_t(B) A) \d t.  
	\end{align}
	By approximation, the left-hand and right-hand sides of \eqref{eq:integral kms} also coincide for all $A,B \in \Afe$.
Then we obtain by dominated convergence for all $A,B \in \Afe_{0,H}$,
	\begin{align*}
		\int_{\IR} f(t) \omega(A \alpha_t(B)) \d t &= \lim_{k\to\infty} \int_{\IR} f(t) \omega_{L_k}(A \alpha^{L_k}_t(B)) \d t =  \lim_{k\to\infty} \int_{\IR} f(t+\i \beta) \omega_{L_k}(\alpha^{L_k}_t(B) A ) \d t \\  &= \int_{\IR} f(t+\i \beta) \omega(\alpha_t(B) A ) \d t.
	\end{align*}
	The KMS property now follows again from  \cite[Prop. 5.3.12]{BR2}. 
\end{proof}
The assumption \eqref{eq:sufficient} is difficult to check, since $\alpha^{L_k}_t(B) \to \alpha_t(B)$ does not converge in operator norm. At least, if one restricts to the particle number preserving case, one can show that it converges in the seminorms or in weighted supremum norms of the form
\begin{align*}
	\nnw{A} := \sup_{n \in \IN} w_n^{-1} \nn{A}_n, \qquad A \in \Afe_0,
\end{align*}
for some sequence $(w_n)_{n \in \IN_0}$ which increases sufficiently fast. However, as such topologies are non-complete, it seems to be difficult to obtain the required equicontinuity for the sequence of states. In order to have a uniform boundedness principle one can try to find a topology on $\Afe_0$ such that it becomes a barrelled space \cite[Ch. 33]{treves} and $\alpha^{L_k}_t(B) \to \alpha_t(B)$ still holds. 

An alternative, technically more involved and open strategy in order to avoid this problem is to use a regularization of the interaction as in \cite{narnhofer1990quantum,narnhofer1991galilei,GebertNachtergaeleReschkeSims.2020,hinrichs2024lieb}. KMS states for such regularized models can be obtained by perturbative arguments, as the regularized interactions are bounded. One can then apply the same Banach-Alaoglu strategy as before for a sequence with vanishing regularization parameter.  Such an approach has recently been proposed in \cite{narnhofer2020local} for the model of \cite{narnhofer1990quantum,narnhofer1991galilei} with smeared-out potentials, using the auto-correlation lower bounds for KMS states \cite[Theorem 5.3.15]{BR2}. The advantage of this characterization is that it allows to use operator inequalities instead of equicontinuity arguments. We leave it as an open question for further work if such an approach also works for the Gaussian regularization of \cite{GebertNachtergaeleReschkeSims.2020,hinrichs2024lieb} or for other smearing functions, which are also admissible in the analysis of \cite{hinrichs2024lieb}.

\section{The structure of the particle number preserving CAR algebra}
\label{sec:structure}

In this section we will describe the structure of the particle-number preserving subalgebra $\Af_0$, in particular, its $n$-particle sectors and the relations between them. We will use its inductive limit property, which was already studied by Bratteli in the 70's \cite{bratteli1972inductive}. It then turns out that we find an analog structure like the one of the bosonic resolvent algebra \cite{bh1}. 

For bounded operators $A_1, \ldots, A_n$ on Hilbert spaces $\H_i$, we can create a bounded operator on the $n$-times antisymmetric tensor product by
\[
A_1 \wedge \ldots \wedge A_n := \frac{1}{n!} \sum_{\sigma \in S_n}  A_{\sigma(1)} \otimes \cdots   \otimes A_{\sigma(n)}  \in \Lb(\H_1 \wedge \ldots \wedge \H_n),
\]
see e.g. \cite{garcia2023symmetric} for a thorough discussion. 
It is easy to check that
\begin{align}
	\label{eq:antisymm operator relation}
	\Pa{n} (A_1 \otimes \ldots \otimes A_n) \Pa{n} = (A_1 \wedge \ldots \wedge A_n) \Pa{n}. 
\end{align}
We will now see that the $n$-particle sectors of $\Af_0$ are given by the following algebra.
\begin{defn}
	Let $\kh_n \subseteq \Lb(\FS_n)$ be the $C^*$-subalgebra generated by elements of the form 
	\[
	C_1 \wedge \cdots \wedge C_m \wedge \Id_{n-m},
	\]
	where $m \leq n$, $C_i$ are compact operators on $\hs$ and $\Id_{n-m} :=  \Id \wedge \cdots \wedge \Id$ with $n-m$ factors. 
	This is equal to the algebra
	\[
	 \bigoplus_{m=0}^n \cK(\cF_m) \wedge \Id_{n-m}.
	\]
\end{defn}

\begin{prop}
	\label{th:restriction is khn}
	We have $\Af_0|_{\FS_n} = \kh_n$. 
\end{prop}
\begin{proof}
	Let $(f_k)_{k \in\IN}$ be an orthonormal basis of $\hs$. Then we know that  $\Af$ is the closure of the linear hull of polynomials in $a(f_i)$ and $a^*(f_i)$ \cite{bratteli1972inductive}. Therefore,
	\[
	\Af_0 =  \overline{ \lin \{ a^*(f_{i_1}) \cdots a^*(f_{i_k}) a(f_{j_1}) \cdots a(f_{j_k}) :  k \in \IN_0,~i_1, \ldots, i_k, j_1, \ldots, j_k \in \IN\} }.
	\]
	Now notice that 
	\[
	 a^*(f_i) a(f_j)|_{\FS_n} = \sum_{k=1}^n \Id \otimes \cdots \otimes \overbrace{\ketbra{f_i}{f_j}}^{k\text{th factor}} \otimes \cdots \otimes \Id,
	\]
	where $\ketbra{f_i}{f_j} \psi := f_i \cs{f_j, \psi}$. This shows $\Af_0|_{\FS_n} \subseteq \kh_n$.
	
	On the other hand, for $m \leq n$, we have 
	\[
	F_1 \wedge \ldots \wedge F_m \wedge  \Id_{n-m} \in \kh_n
	\]
	for all finite rank operators $F_p$ on $\hs$, since each $F_p$ is a linear combination of operators of the form $		\ketbra{f_i}{f_j}$, and
		\[
	\ketbra{f_{i_1}}{f_{j_1}} \wedge \ldots \wedge	\ketbra{f_{i_m}}{f_{j_m}} \wedge \Id_{n-m} = \frac{1}{n!} a^*(f_{i_1}) \cdots a^*(f_{i_m}) a(f_{j_1}) \cdots a(f_{j_m}) |_{\FS_n}.
	\]
	Approximating compact operators by finite-rank operators and taking the closure yields
	\[
	\kh_n \subset \overline{ \Af_0|_{\FS_n} } =  \Af_0|_{\FS_n}.
	\]
	Notice that the last equality follows from the fact that the restriction map
	\[
	\Af_0 \rightarrow \Lb(\H), \quad A \mapsto A|_{\FS_n}
	\]
	is a $*$-homomorphism, so its image is closed \cite[Prop 2.3.1]{BR2}.
\end{proof}
Next, we discuss the relations between different sectors. To this end, let $T_x$, $x\in\IR^d$, denote the translation operators by $x$, i.e., $(T_x f)(y) = f(y+x)$, $f \in L^2(\IR^d)$. Obviously, it holds $T_x f \to 0$, $\abs x \to \infty$, in the weak sense. 
\begin{lem}
	\label{th:r to infinity}
	Let $f_1, \ldots, f_n, g_1, \ldots, g_n \in \hs$. Then we have for any $A \in \Af_0$,
	\begin{align}
	\lim_{\abs x\to\infty} &\cs{a^*(g_1) \cdots a^*(T_x g_n) \Omega, A  a^*(f_1) \cdots a^*(T_x f_n)\Omega } \label{eq:lhs} \\ &=  \cs{a^*(g_1) \cdots a^*(g_{n-1}) \Omega, A  a^*(f_1) \cdots a^*(f_{n-1}) \Omega } \cs{g_n,f_n}. \nonumber
	\end{align}
	In particular, $\nn{A}_{n-1} \leq \nn A_n$ for all $n \in \IN$. 
\end{lem}
\begin{proof}
	First assume that $A$ is a polynomial in an even number of creation and annihilation operators. Then since $\lim_{\abs x\to\infty} \cs{h,T_x f_n} = 0$ for all $h \in \hs$, we can anticommute the term $a^*(T_x f_n)$ through all  creation operators on the right, through $A$ and the  creation operators on the left, so that the left-hand side of \eqref{eq:lhs} equals 
	\begin{align*}
			\lim_{\abs x\to\infty} \cs{a^*(g_1) \cdots a(T_x f_n)  a^*(T_x g_n) \Omega, A  a^*(f_1) \cdots a^*(f_{n-1}) \Omega}.
	\end{align*}
	This coincides with the expression on the right-hand side. By approximation, the result then extends to all $A \in \Af_0$. Finally, the norm inequality follows from the fact that the vectors $a^*(f_1) \cdots a^*(f_{n}) \Omega$, $n \in \IN$, when $(f_n)_{n \in \IN}$ is an orthogonal basis of $\hs$, in turn form an orthogonal basis of $\FS_n$. 
\end{proof}
\begin{lem}
	\label{th:limit key relation}
		Let $n \in \IN$, $f_1, \ldots, f_n, g_1, \ldots, g_n \in \hs$ and set
	\begin{align*}
		\psi^{(n-1)} &:= a^*(f_1) \cdots a^*(f_{n-1}) \Omega, \qquad \ph^{(n-1)} := a^*(g_1) \cdots a^*(g_{n-1}) \Omega,  \\
		\psi^{(n)}(x) &:= a^*(T_x f_n) \psi^{(n-1)}  , \qquad 	\ph^{(n)}(x) := a^*(T_x g_n) \ph^{(n-1)} .
	\end{align*}
	Then for any $m < n$ and $C^{(m)} \in \cK(\FS_m)$,
	\begin{align*}
		\lim_{\abs x\to\infty} \cs{\ph^{(n)}(x), (C^{(m)} \wedge \Id_{n-m} )\psi^{(n)}(x) } = 
			\frac{n-m}{n} \cs{\ph^{(n-1)}, (C^{(m)} \wedge \Id_{n-m-1} ) \psi^{(n-1)}} \cs{g_n,f_n},
	\end{align*}
	and the limit vanishes if $m=n$. 
\end{lem}
\begin{proof}
		Each $C^{(m)}$ can be approximated by linear combinations of operators $C_1 \wedge \ldots \wedge C_m$ where $C_j \in \cK(\hs)$ and such a limit in the above inner product is uniform in $x$. Therefore, it suffices to consider those operators instead of $C^{(m)}$. Moreover, because of \eqref{eq:antisymm operator relation}, it even suffices to consider unsymmetrized operators $C_1 \otimes \ldots \otimes C_m$, $C_j \in \cK(\hs)$, instead of $C^{(m)}$. By the definition of the creation operators \eqref{eq:creation op}, one sees
		\begin{align*}
			a^*(f_1) \cdots a^*(f_n) \Omega = \frac{1}{\sqrt{n!}}\sum_{\sigma\in S_n} (-1)^\sigma f_{\sigma(1)} \otimes \cdots \otimes f_{\sigma(n)}.
		\end{align*}
		 We then compute
		\begin{align}
			&\cs{\ph^{(n)}(x), (C_1 \otimes \ldots \otimes C_m  \otimes \Id_{n-m} )\psi^{(n)}(x) } \nonumber \\ &= \frac{1}{n!} \sum_{\sigma,\tau \in S_n} (-1)^{\sigma\circ\tau} \cs{\gw_{\tau(1)}, C_1 \fw_{\sigma(1)}} \cdots  \cs{\gw_{\tau(m)}, C_m \fw_{\sigma(m)}}  \cs{\gw_{\tau(m+1)}, \fw_{\sigma(m+1)}} \cdots   \cs{\gw_{\tau(n)}, \fw_{\sigma(n)}}, \label{eq:expansion1}
		\end{align}
		where we set 
		\begin{align*}
			\fw_j := \begin{cases}
				f_j &: j < n, \\
				T_x f_n &: j = n,
			\end{cases} \qquad \gw_j := \begin{cases}
				g_j &: j < n, \\
				T_x g_n &: j = n.
			\end{cases} 
		\end{align*}
		Now notice that $C_j T_x \to 0$ strongly , $\abs x\to\infty$, because of compactness and weak convergence of the $T_x$,  and $\lim_{\abs x\to\infty}\cs{\gw_j, \fw_n} = 0$ for any $j < n$ (and vice-versa with $\fw$ and $\gw$ interchanged). This means that only those summands are non-zero where $\sigma^{-1}(n) = \tau^{-1}(n) \in \{ m+1, \ldots, n \}$. Thus, in case of $n=m$, the whole sum vanishes. If $m <n$, we have $n-m$ possibilities for the choice of $\sigma^{-1}(n) = \tau^{-1}(n)$. Without loss of generality assume that $\sigma^{-1}(n) = \tau^{-1}(n) = n$, otherwise change the values of $\sigma(k)$, $\tau(k)$, $k=m+1, \ldots, n$ in the same way, which does not change $(-1)^{\sigma\circ\tau}$.  Then \eqref{eq:expansion1} equals
		\begin{align*}
			&(n-m) \cs{\gw_n,\fw_n} \sum_{\sigma,\tau \in S_{n-1}} (-1)^\sigma(-1)^\tau \cs{\gw_{\tau(1)}, C_1 \fw_{\sigma(1)}} \cdots  \cs{\gw_{\tau(m)}, C_m \fw_{\sigma(m)}}  \\ &\qquad \times \cs{\gw_{\tau(m+1)}, \fw_{\sigma(m+1)}} \cdots   \cs{\gw_{\tau(n-1)}, \fw_{\sigma(n-1)}} \\
			&= (n-m) \cs{g_n, f_n} \frac{(n-1)!}{n!} \cs{ \ph^{(n-1)}, \left( C_1 \otimes \ldots \otimes C_m \otimes \Id_{n-m-1} \right)\psi^{(n-1)} }. \qedhere
		\end{align*}
\end{proof}	
\begin{cor}
	\label{th:n-1 restriction}
	Let $A \in \Af_0$, and write $A|_{\FS_n} = \sum_{m=0}^{n} C^{(m)} \wedge \Id_{n-m}$ with $C^{(m)} \in \cK(\FS_m)$, $m \leq n$. Then we have 
	\begin{align*}
		A|_{\FS_{n-1}} = \sum_{m=0}^{n-1} \frac{n-m}{n} C^{(m)} \wedge \Id_{n-m-1}.
	\end{align*}
\end{cor}
\begin{proof}
	Combining \cref{th:limit key relation} with \cref{th:r to infinity} yields
	\begin{align*}
		\cs{\ph^{(n-1)}, A|_{\FS_{n-1}} \psi^{(n-1)}} =   \sum_{m=0}^{n-1} \frac{n-m}{n} 	\cs{\ph^{(n-1)},\left(  C^{(m)} \wedge \Id_{n-m-1} \right) \psi^{(n-1)}}.
	\end{align*}
	This proves the statement, as the linear span of the vectors $\psi^{(n-1)}$ and $\ph^{(n-1)}$ lies dense in $\FS_{n-1}$. 
\end{proof}
This corollary leads to the correct choice of the coherent maps between different particle sectors.
\begin{defn}\thmenum
	\begin{enumerate}[label=(\arabic*)]
		\item 
We define a map 
\begin{align*}
	\kappa_n \: \kh_n \rightarrow \kh_{n-1}, \qquad \kappa_n\left( \sum_{m=0}^{n} C^{(m)}  \wedge \Id_{n-m} \right)  := \sum_{m=0}^{n-1} \frac{n-m}{n} C^{(m)} \wedge \Id_{n-m-1},
\end{align*}
where $C^{(m)} \in \cK(\FS_m)$, $m = 1, \ldots, n$. 
\item 
Then the inverse limit $\kh$ is defined as the space of sequences $(K_n)_{n \in \IN_0}$, such that $K_n \in \kh_n$, $n \in\IN_0$, $\kappa_n(K_n) = K_{n-1}$ for all $n \in \IN$, and 
\[
\nn{K}_\infty := \sup_{n \in \IN_0} \nn{K_n} < \infty.
\]
It is straightforward to see that $\kh$ equipped with the norm $\nn{\cdot}_\infty$ and pointwise addition, multiplication and adjoint operation is a $C^*$-algebra.
\end{enumerate}
\end{defn}
Moreover, recall that $\Afe_0 \subseteq \Lb(\cF)$ was defined as the closure of $\Af_0$ with respect to the $\nn\cdot_n$ seminorms. 
\begin{prop}
The space $\kh$ equipped with the norm $\nn{\cdot}_\infty$, and pointwise addition, multiplication and adjoint operation is a $C^*$-algebra. Moreover, we have a $*$-isomorphism
\begin{align*}
	\Phi \: \Afe_0 \rightarrow \kh, \qquad A \mapsto (A|_{\FS_n})_{n\in\IN_0}.
\end{align*}
\end{prop}
\begin{proof}
	In the following we use \cref{th:restriction is khn}, i.e., that for each $K_n \in \kh_n$ we can find an $A \in \Af_0$ such that $A|_{\FS_n} = K_n$.
	It is then straightforward to verify that $\kappa_n$ is in fact a $*$-homomorphism, since for $A_1,A_2 \in \Af_0$,
	\begin{align*}
		\kappa_n(A_1|_{\FS_n} A_2|_{\FS_n}) &= A_1 A_2|_{\FS_{n-1}} = \kappa_n(A_1|_{\FS_{n}}) \kappa_n(A_2|_{\FS_{n}}), \\
			\kappa_n(A_1|_{\FS_n} )^* &= (A_1|_{\FS_{n-1}})^* = A_1^*|_{\FS_{n-1}} = \kappa_n(A_1^*|_{\FS_n}).
	\end{align*}
	Furthermore, \cref{th:r to infinity} yields that $\kappa_n$ is continuous with $\nn{ \kappa_n(K_n) } \leq \nn{K_n}$ for $K_n \in \kh_n$. Thus, it is easy to conclude that $\kh$ is a $C^*$-algebra. 
	
 The extension to $\Afe_0$ is well-defined: if $A_m \to A \in \Lb(\FS)$ in the $\nn \cdot_n$ seminorms, then $\sup_{n \in \IN_0} \nn{A|_{\FS_n}} \leq \nn A$ and the coherence condition is satisfied because of the continuity of the $\kappa_n$. 
 	Clearly, $\Phi$ is an injective  $*$-homomorphism. It remains to show the surjectivity. Let $K \in \kh$. Then,  by \cref{th:restriction is khn} and the coherence condition, we find for each $n \in \IN_0$ an $A_n \in \Af_0$ such that $A_n|_{\FS_m} = K_m$ for all $m \leq n$. This means $A_n|_{\FS_m} \to K_m$, $n \to \infty$ and therefore, $A_n$ converges in the seminorms topology to an operator $A \in \Afe_0 \subseteq \Lb(\FS)$ given by $(A\psi)_n := K_n \psi_n$, which satisfies $\Phi(A) = K$. 
\end{proof}
Therefore, from now on we will identify $\Afe_0$ and $\kh$ without further emphasizing it. 
\begin{ex}
	\label{ex:counter}
	We have $\Af_0 \subsetneq \Afe_0$. Let  $(f_k)_{k\in\IN}$  be an orthonormal basis of $\hs$ and let $A_n \in \kh_n$ be given by 
	\begin{align*}
		A_n = \sum_{k=1}^n (-1)^k n(f_1) \cdots n(f_k) |_{\FS_n}.
	\end{align*}
	Then we see that $f_{i_1} \wedge \ldots \wedge f_{i_n}$, $i_1, \ldots, i_n \in \IN$, are eigenvectors of $A_n$ with eigenvalue $0$ or $-1$. Therefore, $\nn{A_n} \leq 1$ for all $n \in \IN$. Since also the coherence relation $\kappa_n(A_n) = A_{n-1}$ is satisfied by construction, we get that $A = (A_n)_{n \in \IN_0} \in \Afe_0$. But $A \not\in \Af_0$: assume that there is $B \in \Af_0$, being a polynomial in $\Id$ and the creation and annihilation operators with $f_j$'s as arguments, such that $\nn{A-B} < \epsilon$ for some small $\epsilon > 0$. We can write without loss of generality
	\[
	B = b \Id + B', 
	\]
	where $b \in \IC$ and $B'$ is a polynomial in creation and annihilation operators in the $f_j$ such that all creation operators are shifted to the right. Let $k \in \IN$ be an even number such that all the indices $j$ of the $f_j$ in the creation operators in $B'$ appear in $1, \ldots, k$. Then $A( f_1 \wedge \ldots \wedge f_k) = 0$, $B' (f_1 \wedge \ldots \wedge f_k )= 0$, and hence, 
	\[
	\nn{(A-B) (f_1 \wedge \ldots \wedge f_k) } = \abs b \nn{f_1 \wedge \ldots \wedge f_k}.
	\]
	 On the other hand, 
	 \[
	 \nn{(A-B) (f_1 \wedge \ldots \wedge f_k \wedge f_{k+1}) } = \abs {1+b} \nn{f_1 \wedge \ldots \wedge f_k\wedge f_{k+1}}.
	 \]
	 In sum, we have $\abs b < \epsilon$ and $\abs{1+b} < \epsilon$, which is a contradiction. 
\end{ex}

\section{Invariance under the dynamics for particle number preserving observables}
\label{sec:invariance}
In this part, we study the dynamics induced by the Hamiltonian \eqref{eq:Hamiltonian} on the  algebra $\Afe_0$. In particular, we prove the particle number preserving version of \cref{th:main} involving $\Afe_0$. 
For the proof we can directly use some results of \cite{bh1} for the general unsymmetrized case and adapt the results from the symmetric to the antisymmetric setting. To this end, let $\FS_n^\otimes := \hs^{\otimes n}$ denote the unsymmetrized tensor product and $\khs_n$ the corresponding unsymmetrized version of $\kh_n$, which is defined as follows. For any set $I \subseteq \{1,\ldots n\}$ we consider a subalgebra of $\Lb(\FS_n^\otimes)$ given by
\begin{align*}
	\kh(I) := \cK_1 \otimes \ldots \otimes \cK_n, \quad \cK_j = \begin{cases}
		\cK(\hs) &: j \in I, \\
		\IC \cdot \Id &: j \not\in I,
	\end{cases}
\end{align*}
where the tensor product denotes the tensor product of $C^*$-algebras. Then $\khs_n$ is defined as
\begin{align*}
	\khs_n := \bigoplus_{I \subseteq \{1,\ldots,n\}} \kh(I).
\end{align*}%
We see that $\kh_n$ is the antisymmetrized version of $\khs_n$, i.e.,
\begin{align}
	\label{eq:khs kh relation}
	\kh_n = \Pa{n} \khs_n  |_{\FS_n},
\end{align}
and, with a slight abuse of notation, we will identify $\kh_n$ from now on with $\Pa{n} \khs_n \Pa{n}$, which is a subalgebra of $\khs_n$.

Furthermore, notice that the components $H_n$ of the Hamiltonian \eqref{eq:Hamiltonian n} can be  defined on $\FSo{n}$  as well as self-adjoint operators. Moreover, we  write
\begin{align*}
	\Vb := \bigoplus_{n=0}^\infty V_n, \qquad V_n = \ssum{i,j=1, \\ i \not= j}^n V_{ij}
\end{align*}
and we will use this notation both on the antisymmetric and non-symmetrized Fock space sectors.  Let $H_n^0$ denote the free Hamiltonian without the two-body interaction, i.e.,  \eqref{eq:Hamiltonian n} with $V = 0$. Let $\An \in \Lb(\FSo{n})$ or $\An \in \Lb(\FS_n)$.  We set $\alpha^0_t(\An) :=  e^{\i t H_n^0} \An e^{-\i t H_n^0}$  and $\gamma_{t} :=\alpha_{t} \circ \alpha^0_{-t}$. The latter can be written as a Dyson series,
\begin{align}
	\label{eq:dyson}
\gamma_{t}(\An) &= \sum_{l=0}^\infty D_{n,l}(t)(\An),
\end{align}
where $D_{n,0}(t)(\An) := \An$ and, for $l \geq 1$,
\begin{align}
	\label{eq:Dl defn}
	  D_{n,l}(t)(\An) &:=  (-\i)^l \int_0^t \int_0^{s_l} \cdots \int_0^{s_2} [\ldots[\An,V_n(s_1)] \ldots, V_n(s_l) ] \d s_1 \ldots \d s_l,
\end{align}
with $V_n(s) := \alpha^0_{s}(V_n)$, $s \in \IR$. The integrals are to be understood in terms of the strong operator topology (as are all operator-valued integrals in the following). Introducing
\begin{align}
	\nonumber
\delta_n(t)(\An) &:= - \i [\An, V_n(t) ], \\
\label{eq:Deltan defn}
\Delta_n(t)(\An) &:= \int_0^t \delta_n(s)(\An) \d s \in \Lb(\FSo{n}),
\end{align}
we can also write \eqref{eq:Dl defn} recursively in the form
\begin{align*}
	D_{n,1}(t)(\An) = \Delta_n(t)(\An), \qquad D_{n,l}(t)(\An) =  \int_0^t \delta_n(s)( D_{n,l-1}(s)(\An) ) \d s. 
\end{align*}

In \cite{bh1}, the following three lemmas were proven, cf. Lemma 4.1, 4.2 and 4.3 therein. 
\begin{lem}
 For all $n \in \IN_0$ and $t \in \IR$, $\Delta_n(t)$ maps $\khs_n$ into itself. It is pointwise norm continuous in $t$ and bounded by $\nn{\Delta_n(t)(\An)}_n \leq 2 \abs t \nn{V_n} \nn{\An}$, $\An \in \khs_n$. 
\end{lem}
\begin{lem}
	\label{th:approximate Dl integral}
Let $D \: \IR \rightarrow \khs_n$ be norm-continuous. Then, for all $t \in \IR$, 
\begin{enumerate}[label=(\alph*)]
	\item 
 $\IR \rightarrow \Lb(\FSo{n})$, $s \mapsto \delta_n(s)(D(s))$ is continuous in the strong operator topology, 
 \item $\int_0^t \delta_n(s)(D(s)) \d s \in \khs_n$, 
 \item the approximation 
\begin{align*}
	\int_0^t \delta_n(s)(D(s)) \d s = \lim_{k\to\infty} \sum_{j=1}^k (\Delta_n(jt/k) - \Delta_n((j-1)t/k))(D(jt/k))
\end{align*}
holds in the sense of norm convergence,
\item  $\IR \rightarrow \Lb(\FSo{n})$, $t \mapsto 	\int_0^t \delta_n(s)(D(s)) \d s$ is norm-continuous, and
\item $\nn{ \int_0^t \delta_n(s)(D(s)) \d s} \leq  2 \nn{V_n} \int_0^{\abs t} \nn{D(s)} \d s$.
\end{enumerate}
\end{lem}

\begin{lem}
	\label{th:unsymmetrized main result}
	For each $n \in \IN_0$ and $t \in \IR$, the Dyson maps $\gamma_t$ map $\khs_n$ onto itself, i.e., they are automorphisms of $\khs_n$, and the functions
	\[
	\IR \rightarrow \Lb(\FSo{n}), \quad t \mapsto \gamma_t(\An)
	\]
	are norm-continuous for all $\An \in \Lb(\FSo{n})$.
\end{lem}
Similarly as for the bosonic case \cite[Prop. 4.4]{bh1}, we can restrict the statement of \cref{th:unsymmetrized main result} to the  antisymmetric setting. Remember that we identify $\kh_n$ with $\Pa{n} \khs_n \Pa{n}$.
\begin{cor}
	\label{th:khn invariance}
	We have $\Delta_n(t)(\kh_n) = \kh_n$, $\alpha_t(\kh_n) = \kh_n$, and the function $t \mapsto \alpha_t(\An)$ is norm continuous for all $\An \in \kh_n$. 
\end{cor}
\begin{proof}
	Both of the operators $H_n$ and $V_n$ commute with $\Pa{n}$. Therefore, using the definition of $\gamma_t$, we obtain $\gamma_t(\Pa{n} \An \Pa{n}) = \Pa{n} \gamma_t(\An) \Pa{n}$ for all $\An \in \khs_n$. This and \cref{th:unsymmetrized main result} imply that $\gamma_t(\kh_n) = \kh_n$, and therefore, $\alpha_t (\kh_n) = \alpha^0_t (\gamma_{-t}(\kh_n)) = \kh_n$. The continuity follows from the continuity of $ \alpha^0_t$ and the continuity statement in \cref{th:unsymmetrized main result}.
	\end{proof}

\begin{lem}
	\label{th:commutator clustering}
	Let $n \in \IN$, $f_1, \ldots, f_n, g_1, \ldots, g_n \in \hs$ and $t \in \IR$. Assume that either
	\begin{enumerate}[label=(\alph*)]
		\item \label{it:ahDeltaash} $B_m = a(h_1) \int_0^t [a^*(h_2),\Vb(s)] \d s|_{\FS_m}$ with $h_1,h_2 \in \hs$, or
		\item \label{it:VntA} $B_m = [V_m(t), A|_{\FS_m}]$ with $A \in \Af_0$,
	\end{enumerate}
	for $m \in \{n-1,n\}$.
	Then we have
	\begin{align*}
		\lim_{s\to\infty} &\cs{a^*(g_1) \cdots  a^*(g_{n-1}) a^*(T_x g_n) \Omega, B_n  a^*(f_1) \cdots a^*(f_{n-1}) a^*(T_x f_n)\Omega } \\ &=  \cs{a^*(g_1) \cdots a^*(g_{n-1}) \Omega, B_{n-1}  a^*(f_1) \cdots a^*(f_{n-1}) \Omega } \cs{g_n,f_n} .
	\end{align*}
\end{lem}
\begin{proof}
	We will use the same strategy as in the proof of \cref{th:r to infinity}, i.e., we have to show that the commutator between $ B_n$ and $a^*(T_x f_n)$ vanishes in the limit $\abs x \to \infty$. 
	
	First consider the case \ref{it:ahDeltaash}, assume that $V$ is compactly supported in $\B_{r_V}$, the closed ball around the origin of radius $r_V$, $f_n$ is compactly supported, and $h \in \hs$ is supported in a ball of radius $r_h$.  If we show that $[\Vb,a^\#(h)]$ anticommutes with every $a^*(f)$ with $\supp f \subseteq \B_{r_h + r_V}^c$, then the statement follows for the given assumptions by anticommuting $a^*(T_x f_n)$ to the vacuum on the left-hand side. In doing so we use that $T_x f_n$ is supported outside of any ball for $\abs x$ big enough. 
	
	For the case \ref{it:VntA}, we first assume that $A$ is a polynomial of creation and annihilation operators of functions having compact support. Then we can expand the commutator $[V_n(t),A_n]$ into the single creation and annihilation operators and proceed as in the previous case. 
		
 Under the general assumptions the statement follows from the approximation of $V$, $f_n$, $h$ and the functions in the creation and annihilation operators of $A$  with compactly supported functions in norm. Here we use that  the convergence is independent of $x$ and the translations $T_x$ commute with  the free time evolution. 
	
	It remains to show that $[\Vb,a^\#(h)]$ anticommutes with every $a^*(f)$ when $\supp f \subseteq \B_{r_h + r_V}^c$. 
	 The $n$-particle components of the commutators between $\Vb$ and the creation and annihilation operators are given by
	\begin{align*}
([\Vb, a^*(h)] \psi )_n(x_1, \ldots, x_n) &= \frac{2}{n^{1/2}}  \sum_{k=1}^{n} (-1)^{k+1} h(x_k) \ssum{i=1,\\i\not=k}^n V(x_i- x_k) \psi_{n-1}(x_1, \ldots, \widehat{x_k}, \ldots, x_n), \\
([\Vb, a(h)] \psi )_n(x_1, \ldots, x_n) &= - 2\sqrt{n+1}  \sum_{k=1}^n \int \overline{h(x)}   V(x-x_k)  \psi_{n+1}(x,x_1,\ldots, x_n) \d x,
	\end{align*}
which follows directly from the definitions \eqref{eq:creation op} and \eqref{eq:annihilation op}. Using this, we compute
	\begin{align*}
		&(a^*(f) [\Vb,a(h)] \psi)_n(x_1, \ldots, x_n) = n^{-1/2} \sum_{i=1}^{n} (-1)^{i+1} f(x_i) ([\Vb,a(h)] \psi)_{n-1}(x_1, \ldots, \widehat{x_i}, \ldots, x_n) \\
		& = -2 \sum_{i=1}^{n} (-1)^{i+1} f(x_i) \ssum{k=1,\\k \not= i}^n \int \overline{h(x)}   V(x-x_k) , \psi_{n}(x,x_1,\ldots, \widehat{x_i}, \ldots x_n) \d x  \\
		&( [\Vb,a(h) ] a^*(f)\psi)_n(x_1, \ldots, x_n) = - 2\sqrt{n+1}  \sum_{k=1}^n \int \overline{h(x)}   V(x-x_k)  (a^*(f)\psi)_{n+1}(x,x_1,\ldots, x_n) \d x \\
		& = - 2 \sum_{k=1}^n \int \overline{h(x)}   V(x-x_k) 
		\left(f(x) \psi_n(x_1, \ldots, x_n) + \sum_{i=1}^{n} (-1)^i f(x_i) \psi_n(x,x_1, \ldots, \widehat{x_i}, \ldots, x_n) \right)  \d x,
\end{align*}
\allowdisplaybreaks
from which it follows
	\begin{align}
		\{ a^*(f), [\Vb,a(h)] \} \psi_n (x_1, \ldots, x_n)  &= -2  \sum_{k=1}^n \int \overline{h(x)}   V(x-x_k)  \bigg(
		f(x) \psi_n(x_1, \ldots, x_n) \label{eq:integrand1}  \\  &\qquad + (-1)^k f(x_k) \psi_n(x, x_1, \ldots, \widehat{x_k}, \ldots, x_n) \bigg) \d x. \label{eq:integrand2}
	\end{align}
	For the integrand \eqref{eq:integrand1} to be non-zero, we need to have $\abs x < r_h$ and $\abs x \geq r_h + r_V$, which is a contradiction. Likewise, for the one in \eqref{eq:integrand2} to be non-zero we  need $\abs x < r_h$, $\abs{x-x_k} < r_V$ and $\abs{x_k} \geq r_h + r_V$, which again yields a contradiction. 
	Similarly, a direct computation shows
	\begin{align*}
		&a^*(f) [\Vb,a^*(h)] \psi_n(x_1, \ldots, x_n) =   n^{-1/2} \sum_{i=1}^n (-1)^{i+1} f(x_i)  ([\Vb,a^*(h)] \psi)_{n-1}(x_1, \ldots, \widehat{x_i}, \ldots  ,x_n) \\ 
		&\qquad =  2 (n(n-1))^{-1/2} \sum_{i=1}^n (-1)^{i+1} f(x_i) \ssum{k=1,\\ k\not=i}^{n} (-1)^{k+1} \sgn(i-k) h(x_k) \\  
		&\qquad \qquad \times \ssum{j=1,\\ j \not= i,k}^n V(x_j - x_k)  \psi_{n-2}(x_1, \ldots, \widehat{x_i}, \ldots, \widehat{x_k}, \ldots, x_n), \\
		&[\Vb,a^*(h)] a^*(f)  \psi_n(x_1, \ldots, x_n) \\  &\qquad = 	2 n^{-1/2} \sum_{k=1}^{n} (-1)^{k+1} h(x_k) \ssum{j=1,\\ j\not= k}^n V(x_j- x_k) (a^*(f) \psi)_{n-1}(x_1, \ldots, \widehat{x_k}, \ldots, x_n) \\
		&\qquad = 2 (n(n-1))^{-1/2} \sum_{k=1}^{n} (-1)^{k+1} h(x_k) \ssum{j=1,\\ j\not= k}^n V(x_j- x_k) \\  &\qquad \qquad \times \ssum{i=1,\\i\not= k}^n (-1)^{i+1} \sgn(k-i) f(x_i)  \psi_{n-2}(x_1, \ldots, \widehat{x_i},\ldots, \widehat{x_k}, \ldots, x_n), 
	\end{align*}
	which implies
	\begin{align*}
		&\{a^*(f), [\Vb,a^*(h)] \} \psi_n(x_1, \ldots, x_n) = 2 (n(n-1))^{-1/2} \\ &\qquad \times \sum_{i=1}^n \ssum{k=1,\\ k\not=i}^{n}  (-1)^{i+1} f(x_i) (-1)^{k+1} \sgn(i-k) h(x_k) V(x_i - x_k)  \psi_{n-2}(x_1, \ldots, \widehat{x_i}, \ldots, \widehat{x_k}, \ldots, x_n).
	\end{align*}
	We see again that this is zero for	$\supp f \subseteq \B_{r_h + r_V}^c$, since each summand can only be non-zero if $\abs{x_k} < r_h$ and $\abs{x_i - x_k} < r_V$, i.e., if $\abs{x_i} < r_h + r_V$.
\end{proof}

\begin{lem}
	\label{th:kappa Delta compatibility}
For all $n \in \IN$, we have
	\begin{align*}
		\kappa_n \circ \Delta_n(t) = \Delta_{n-1}(t) \circ \kappa_n. 
	\end{align*}
	In particular, for all $A \in \Af_0$ and $n \in \IN$,
	\[
	(\Delta_m(t)(A|_{\FS_m}))_{m=0}^n \in \Af_0|_{\FS_{\leq n}}.
	\]
\end{lem}
\begin{proof}
	Let $f_1, \ldots, f_n, g_1, \ldots, g_n \in \hs$, and set $\psi = a^*(g_1) \cdots a^*(g_{n-1}) \Omega$, $\varphi =a^*(f_1) \cdots a^*(f_{n-1})\Omega$. Let $A \in \Af_0$ and write $A_m := A|_{\FS_m}$ for all $m$. Then we find, noting that $T_x$ and $e^{-\i s (-\Delta)}$ commute, 
	\begin{align*}
		\cs{ a^*(T_x g_n) \psi, [V_n(s),A_n] a^*( T_x f_n) \varphi} &= \cs{e^{-\i s H^0} a^*(T_x g_n) \psi, [V_n,\alpha^0_{-s}(A_n)] e^{-\i s H^0} a^*(T_x f_n) \varphi} \\
		&= \cs{a^*( e^{-\i s (-\Delta)} T_x  g_n) \psi_s, [V_n,\alpha^0_{-s}(A_n)]  a^*(e^{-\i s (-\Delta)} T_x f_n) \varphi_s},
	\end{align*}
	where $\psi_s = a^*(e^{-\i s (-\Delta)} g_1) \cdots a^*(e^{-\i s (-\Delta)} g_{n-1}) \Omega$ and $\varphi_s =a^*(e^{-\i s (-\Delta)} f_1) \cdots a^*(e^{-\i s (-\Delta)} f_{n-1})\Omega$. So we obtain with \cref{th:commutator clustering} that
	\begin{align*}
		\lim_{\abs x\to\infty}  \cs{  a^*( T_x g_n) \psi, [V_n(s),A_n] a^*( T_x f_n) \varphi} = \cs{ \psi, [V_{n-1}(s),A_{n-1}]  \varphi} \cs{g_n,f_n}.
	\end{align*}
	By integration from $s=0$ to $t$ and dominated convergence, we find 
	\begin{align*}
		\lim_{\abs x\to\infty}  \cs{  a^*( T_x g_n) \psi, \left[\int_0^t V_n(s) \d s,A_n\right] a^*( T_x f_n) \varphi} = \cs{ \psi,  \left[\int_0^t V_{n-1}(s) \d s,A_{n-1} \right]  \varphi} \cs{g_n,f_n} .
	\end{align*}
	Since we know that $\Delta_n(t)(A_n) \in \kh_n$ by \cref{th:khn invariance}, we can use \cref{th:limit key relation}, which yields
	\begin{align*}
		\cs{\psi, \kappa_n(\Delta_n(A_n)) \varphi} \cs{g_n,f_n} &= 	\lim_{\abs x\to\infty}  \cs{  a^*( T_x g_n) \psi, \Delta_n(A_n) a^*( T_x f_n) \varphi} 
		  \\ &= \cs{ \psi, \Delta_{n-1}(t)(A_{{n-1}}) \varphi} \cs{g_n,f_n}
		  = \cs{\psi, \Delta_{n-1}(t)(\kappa_n(A_n)) \varphi}  \cs{g_n,f_n}.
	\end{align*}
	This implies $\kappa_n \circ \Delta_n(t) = \Delta_{n-1}(t) \circ \kappa_n$, since for any $B_n \in \kh_n$ we can find a $A \in \Af_0$ such that $A_n = B_n$. 
	Furthermore,
	\[
	\kappa_n( \Delta_n(t)(A_n )) = \Delta_{n-1}(t)(\kappa_n(A_n)) = \Delta_{n-1}(t)(A_{{n-1}}).
	\]
	So $(\Delta_n(t)(A_n))_{n\in \IN_0}$ forms a coherent sequence, which proves the last statement of the lemma. 
\end{proof}

\begin{prop}
	\label{th:kappa alpha compatibility}
	For all $n \in \IN$, we have
	\begin{align}
		\label{eq:kappa n compatibility}
		\kappa_n \circ \alpha_t = \alpha_{t} \circ \kappa_n. 
	\end{align} 
\end{prop}
\begin{proof}
	The proof works in the same way as in \cite[Lemma 4.5]{bh1}. First one can directly show \eqref{eq:kappa n compatibility} if we replace $\alpha_t$ by the free time evolution $\alpha^0_t$. To this end, notice that the latter does not mix tensor factors and the conjugation of compact operators with the unitary one-body time-evolution is again compact. Thus, $\alpha^0_t$ keeps the tensor product structure of $\kh_n$ invariant and \eqref{eq:kappa n compatibility} for the free time-evolution follows from the definition of $\kappa_n$.
	
	It therefore suffices to prove \eqref{eq:kappa n compatibility} with $\alpha_t$ replaced by $\gamma_t$. Then we can use the Dyson series \eqref{eq:dyson} and show by induction over $l$ that $\kappa_n(D_{n,l}(\An)) = D_{n-1,l}(\kappa_n(\An))$ for all $l \in \IN_0$ and all $\An \in \kh_n$. For $l = 0$, this is the statement of \cref{th:kappa Delta compatibility}. For the induction step, we use
	\cref{th:approximate Dl integral} and the norm continuity of $\kappa_n$ in order to obtain
	\begin{align*}
		\kappa_n(D_{n,l}(t)(\An)) &= \kappa_n \left( \int_0^t \delta_n(s)(D_{n,l-1}(s)(\An)  ) \d s \right) \\
		&= \lim_{k\to\infty}	\kappa_n \left( \sum_{j=1}^k ( \Delta_n(jt/k) - \Delta_n((j-1)t/k) )(D_{n,l-1}(jt/k)(\An))  \right) \\
		&=  \lim_{k\to\infty} \sum_{j=1}^k ( \Delta_{n-1}(jt/k) - \Delta_{n-1}((j-1)t/k) )( 	\kappa_n(D_{n,l-1}(jt/k)(\An)))  \\
		 &= \int_0^t \delta_{n-1}(s)(\kappa_n(D_{n,l-1}(s)(\An))) \d s
		  \\
		 &= \int_0^t \delta_{n-1}(s)(D_{n-1,l-1}(s)(\kappa_n(\An))) \d s = D_{n-1,l}(t)(\kappa_n(\An)),
	\end{align*}
	where we assumed that the statement holds for $l-1$. Thus, we find $\kappa_n \circ \gamma_t = \gamma_{t} \circ \kappa_n$, which finishes the proof.
\end{proof}
\begin{proof}[Proof of \cref{th:main}, particle-number preserving case]
	From \cref{th:khn invariance,th:kappa alpha compatibility} we obtain that $\alpha_t(A)|_{\FS_{\leq n}} \in \Af_0|_{\FS_{\leq n}}$ for all $A \in \Afe_0$ and $n \in \IN_0$, so $\alpha_t(A) \in \Afe_0$. Furthermore, \cref{th:khn invariance} also yields the desired continuity in the seminorms. 
\end{proof}
\begin{proof}[Proof of \cref{th:main coro},particle-number preserving case]
	 Consider $\int f(s) \alpha_s(A) \d s \in \Lb(\FS)$, $f \in \Cci(\IR)$, $A \in \Afe_0$. As $\int f(s) \alpha_s(A) \d s|_{\FS_n}$ exists in the norm sense on $\FS_n$ by \cref{th:khn invariance} for all $n \in \IN_0$, Riemann sums of $\int f(s) \alpha_s(A) \d s$ (being elements of $\Afe_0$) converge to this integral with respect to each seminorm $\nn{\cdot}_n$. Since $\Afe_0 \subseteq \Lb(\FS)$ is closed in the seminorms topology, we have $\int f(s) \alpha_s(A) \d s \in \Afe_0$. Let $\Afe_{0,H}$ be the $C^*$-subalgebra of $\Afe_0$ generated by $\int f(s) \alpha_s(A) \d s \in \Lb(\FS)$, $f \in \Cci(\IR)$, $A \in \Afe_0$. The estimate
	\begin{align}
		\label{eq:pettis estimate}
		\nn{\int f(s) \alpha_s(A) \d s} \leq \nn A \int \abs{f(s)} \d s
	\end{align}
 implies that $t \mapsto \alpha_t\left( \int f(s) \alpha_s(A) \d s\right)$ is continuous in norm, which in turn yields that $t \mapsto \alpha_t(B)$ is continuous in norm for all $B \in \Afe_{0,H}$. Furthermore,
	\begin{align}
		\label{eq:delta_k estimate}
		\nn{ \int \delta_k(s) \alpha_s(A) \d s - A }_n \leq \int \delta_k(s) \nn{ \alpha_s(A) - A}_n \d s, \quad n \in \IN_0,
	\end{align}
	shows that $\int \delta_k(s) \alpha_s(A) \d s$ converges to $A$ in the seminorms for a compactly supported continuous Dirac sequence $(\delta_k)_{k \in \IN}$. This proves that $\Afe_{0,H}$ is dense in $\Afe_0$. 
\end{proof}

\section{Full CAR algebra}
\label{sec:full}
In this part, we generalize the content of the previous section to the extension of the full algebra $\Afe$ as it was defined in \eqref{eq:Afe characterization}. We start by observing that $\Afe$ is actually a well-defined $C^*$-algebra.

\begin{prop}
	\label{th:oplus algebra}
	The space $\bigoplus_{k \in \IZ} \Afe_k$ is a well-defined $*$-algebra and hence, $\Afe$ as its closure is a $C^*$-algebra. 
\end{prop}
\begin{proof}
	By definition, we have $\Af_k^* = \Af_{-k}$ and $\Af_k \Af_l \subseteq \Af_{k+l}$ for all $k,l \in \IZ$. Furthermore, we see that $\PN_{n+k} A = A \PN_n$ for $A \in \Af_k$. Thus,
	if $\Af_k \ni A_m \overset{m\to\infty}\to A$ in the seminorm $\nn \cdot_n$ and  $\Af_l \ni B_m \overset{m\to\infty}\to B$ in the seminorm $\nn \cdot_n$, then $\nn{ (A_m^*-A^*) \PN_{n+k} } = \nn{\PN_{n+k} (A_m-A)    } =  \nn{(A_m-A) \PN_{n}  }$, so $A_m^* \PN_{n+k} \overset{m\to\infty}\to A^* \PN_{n+k}$, and
	$B_m A_m \PN_n = B_m \PN_{n+k} A_m \PN_n \overset{m\to\infty}\to   B \PN_{n+k} A \PN_n = B A \PN_n$ in operator norm. This shows $\Afe_k^* = \Afe_{-k}$,  $\Afe_k \Afe_l \subseteq \Afe_{k+l}$, and therefore the first statement of the proposition. Then the second statement follows immediately.
\end{proof}

Let $\Lb(\H_1, \H_2)$ denote the bounded operators from $\H_1$ to $\H_2$. Set $\Vb(s) := \alpha^0_s(\Vb)$ and $V_{ij}(s) := \alpha^0_s(V_{ij})$ for all $s \in \IR$ and $i,j =1, \ldots, n$.  The first step we want to take is to show that $a(f) [\Vb,a^*(f)] \in \Afe_0$, which we will prove initially on the larger unsymmetrized space. The key assumption that leads to compactness here is that the pair potential $V$ is in $C_0(\IR^d)$. This is reflected by the following lemma, which is similar to \cite[Lemma A.8]{bh2}. 
\begin{lem}
	\label{th:12compact}
	Let $0 \not= f \in\hs$ and $E_f := \frac{ \ketbra{f}{f} }{\nn f^2}$ denote the orthogonal projection to the one-dimensional space spanned by $f$. Then the operator
	\begin{align}
		\label{eq:op to converge}
\int_0^t V_{12}(s) \d s(E_f  \otimes \Id),
\end{align}
	acting on $\FSo{2} = \hs \otimes \hs$, is compact for all $t \in \IR$. 
\end{lem}
\newcommand{\x}{\mathbf{x}}
\newcommand{\p}{\mathbf{p}}
\begin{proof}
	First assume  $V$ is compactly supported and $f$ is a Schwartz function. For $m \in \IN$, set $\chi_m := \chr_{\{x \in \IR^d : \abs x \leq m\}}$, and  
	\begin{align*}
		V^m(x_1, x_2) := \chi_m(x_1) V(x_1 - x_2) , \qquad x_1, x_2 \in \IR^d,
	\end{align*}
	which is  compactly supported on $\IR^{2d}$.  Let $\x_j, \p_j$, $j \in \{1,2\}$, denote the vectors of position and momentum operators with respect to the $j$-th component. A direct calculation shows $\alpha^0_s(\x_j) = \x_j + 2 s \p_j$ for all $s \in \IR$, so we find
	\begin{align*}
	\alpha^0_s( V^m(\x_1, \x_2 )  )  = V^m( \x_1 + 2 s \p_1 , \x_2 + 2 s \p_2).
	\end{align*}
	Since $V^m$ is compactly supported, $ V^m(\x_1, \x_2 )  V^m(\p_1, \p_2 )$ is a Hilbert-Schmidt operator \cite[Theorem XI.20]{ReedSimon.1979}  and therefore compact. As long as the linear map $(x,p) \mapsto (x + 2s p, x + 2s' p)$ is invertible (namely when 
	 $s \not= s'$), we conclude that the operator 
	\begin{align*}
	\alpha^0_s( V^m(\x_1, \x_2 )  ) 	\alpha^0_{s'}( V^m(\x_1, \x_2 )  )    = V^m( \x_1 + 2 s \p_1 , \x_2 + 2 s \p_2)  V^m( \x_1 + 2 s' \p_1 , \x_2 + 2 s' \p_2)
	\end{align*}
	is unitarily equivalent to  $ V^m(\x_1, \x_2 )  V^m(\p_1, \p_2 )$ and therefore also compact. It follows that
	\begin{align*}
		\left( \int_0^t \alpha^0_s(V^m(\x_1, \x_2))  \d s \right)^2 = \int_0^t \int_0^t  \alpha^0_s(V^m(\x_1, \x_2))   \alpha^0_{s'}(V^m(\x_1, \x_2))   \d s \d s'
	\end{align*}
	is compact, and thus also $ \int_0^t \alpha^0_s(V^m(\x_1, \x_2))  \d s $ by polar decomposition. Writing $f_{-s} := e^{-\i s (-\Delta)} f$, we see that
	\begin{align}
		\label{eq:integrand bounded}
\int_0^t \alpha^0_s(V^m(\x_1, \x_2))  \d s (E_f \otimes \Id) =    \int_0^t \alpha^0_s(  	 V^m(\x_1, \x_2)    (E_{f_{-s}} \otimes \Id)   )  \d s ,
	\end{align}
which is compact as well.  

 We now show that the operator \eqref{eq:integrand bounded} converges to the operator \eqref{eq:op to converge} as $m \to \infty$ in operator norm, implying that \eqref{eq:integrand bounded} must be compact as well. To this end, note that for all $s \in \IR$, $f_{-s}$ is a Schwartz function, so $\chi_m f_{-s} \overset{m\to\infty}{\to}   f_{-s}$ in $L^2$ norm and therefore we have pointwise convergence for fixed $s$,
	\begin{align*}
  V^m(\x_1, \x_2)    (E_{f_{-s}} \otimes \Id)  = V_{12} \left( \frac{\ketbra{\chi_m f_{-s}}{f_{-s}}}{\nn f^2}  \otimes \Id \right)  \overset{m\to\infty}{\to}   V_{12} \left( \frac{\ketbra{ f_{-s}}{f_{-s}}}{\nn f^2}  \otimes \Id\right)  = V_{12} (E_{f_{-s}} \otimes \Id),
	\end{align*}
		in operator norm. As the operator norm of the integrand in \eqref{eq:integrand bounded} is uniformly bounded in $s$ and $m$, the whole integral converges in operator norm.
	
	 If $V \in C_0(\IR^d)$ is not compactly supported, we can approximate $V$ with respect to the supremum norm with continuous compactly supported functions. Since the integral in \eqref{eq:op to converge} is norm continuous with respect to the supremum norm of $V$, we can approximate this operator in operator norm by the same integrals with compactly supported pair potentials. Thus,  \eqref{eq:op to converge} is also compact in this case. Similarly, we can approximate arbitrary $f \in \hs$ by Schwartz functions with respect to the $L^2$ norm, so that the corresponding projections converge in operator norm. 
\end{proof}
Next, recall the definition of the creation and annihilation operators $a^*_n(f) \in \Lb(\FSo{n},\FSo{n+1})$, $n \in \IN_0$, and $a_n(f)  \in \Lb(\FSo{n},\FSo{n-1})$, $n \in \IN$, $f \in \hs$, on the unsymmetrized Fock space \cite{BR2}:
\begin{align*}
	a^*_n(f)(f_1 \otimes \cdots \otimes f_n) &= \sqrt{n+1} f \otimes f_1 \otimes \cdots \otimes f_n, \qquad f_j \in \hs, \\
	a_n(f)(f_1 \otimes \cdots \otimes f_n) &= \sqrt{n} \cs{f,f_1} f_2 \otimes \cdots \otimes f_n.
\end{align*}
Particularly, their restriction to the antisymmetric Fock space coincides with the given definitions \eqref{eq:creation op} and \eqref{eq:annihilation op}, i.e., $a^\#_n(f)|_{\FS_n} = a^\#(f)|_{\FS_n}$. Similarly to \eqref{eq:Deltan defn}, we will write $\Delta_n(t)(A) = \i \int_0^t [A,\Vb(s)] \d s$ for $A \in \Lb(\FS_n, \FS_{n+1})$ or $A \in \Lb(\FSo{n}, \FSo{n+1})$. 
\begin{lem}
	\label{th:is in khn}
	For all $n \in \IN$ and $f \in \hs$, $a(f) \Delta_n(t)(a^*(f)|_{\cF_n}) \in\kh_n$. 
\end{lem}
\begin{proof}
	We have, evaluated as operators on $\FSo{n}$,
	\begin{align}
		a_{n+1}(f) \Delta_n(t)(a^*_n(f)) &= \i a_{n+1}(f) a^*_n(f) \int_0^t V_n(s) \d s - \i a_{n+1}(f) \int_0^t V_{n+1}(s) \d s a^*_n(f)   \nonumber \\
		&= 	\i  a_{n+1}(f) a^*_n(f)   \int_0^t  \sum_{\substack{i,j=1, \\ i\not= j}}^n   V_{ij}(s) \d s - \i a_{n+1}(f) \int_0^t   \sum_{\substack{i,j=0, \\ i\not= j}}^{n}  V_{ij}(s) \d s a^*_n(f) \nonumber  \\
		&= 	 -  2 \i a_{n+1}(f)  \sum_{\substack{i=1}}^{n}  \int_0^t   V_{0i}(s) \d s a^*_n(f)   \nonumber \\
		&= 	 -  2 \i \sum_{\substack{i=1}}^{n}  a_{n+1}(f)  \int_0^t   V_{0i}(s) \d s (E_f \otimes \Id_{\FSo{n}}) a^*_n(f)  , \label{eq:anDeltaan}
	\end{align}
	where the subindex $0$ refers to the tensor factor additionally attached by $a^*_n(f)$.
	Each summand in \eqref{eq:anDeltaan} only acts on the $i$-th tensor factor, i.e., it is of the form
	\[
	\Id \otimes \cdots  \otimes \Id \otimes K_i \otimes \Id \otimes \cdots \otimes \Id,
	\]
	where the operator $K_i$ stands at the $i$-th position and is given in the weak form by
	\begin{align*}
		\cs{h, K_i g} &=   \cs{f \otimes h,  \int_0^t V_{12}(s) \d s   (E_f  \otimes \Id) f \otimes g }, \qquad g,h \in \hs.
	\end{align*}
	As $ \int_0^t V_{12}(s) \d s   (E_f  \otimes \Id) $  is compact by \cref{th:12compact}, also $K_i$ needs to be compact. 
\end{proof}
\cref{th:is in khn} tells us that all the $n$-particle sectors $a(f) \Delta_n(t)(a^*(f)|_{\cF_n})$ are in the correct algebra. Next, the coherence relations between them need to be checked. 
\begin{prop}
	\label{th:Deltat af restriction}
	For all $n \in \IN$, we have $ (a(f)\Delta_m(t)(a^*(f)|_{\FS_{m}}))_{m=0}^n \in \Af_0|_{\FS_{\leq n}}$.
\end{prop}
\begin{proof}
	As in \cref{th:kappa Delta compatibility}, we have to show that, for all $n \in \IN$,
	\begin{align*}
		\kappa_n ( a(f) \Delta_n(t)(a^*(f)|_{\cF_n} )) = a(f) \Delta_{n-1}(t)(a^*(f)|_{\cF_{n-1}}).
	\end{align*}
	
	Let $\psi = a^*(g_1) \cdots a^*(g_{n-1}) \Omega$ and $\varphi =a^*(f_1) \cdots a^*(f_{n-1})\Omega$ for  $f_1, \ldots,f_n,g_1, \ldots, g_n \in \hs$. \cref{th:limit key relation} and \cref{th:commutator clustering} \ref{it:ahDeltaash} directly yield
	\begin{align*}
		\cs{\psi, \kappa_n (a(f) \Delta_n(t)(a^*(f)|_{\FS_{n}})) \varphi} \cs{g_n,f_n} &= \lim_{\abs x\to\infty}\cs{ a^*(T_x g_n) \psi, a(f) \Delta_n(t)(a^*(f)|_{\cF_n} )  a^*( T_x f_n) \varphi} \\ &= \cs{\psi, a(f) \Delta_{n-1}(t)(a^*(f)|_{\FS_{n-1}}) \varphi} \cs{g_n,f_n}. \qedhere
	\end{align*}
\end{proof}

\begin{lem}
	\label{th:approx deltas Daf integral}
	Let $n \in \IN_0$ and $s \mapsto D(s) \in \Lb(\FS_n, \FS_{n+1})$ be continuous in norm. Then we have that
	\begin{align*}
		-\i \int_0^t  [D(s),\Vb(s)] \d s \big|_{\cF_n} =  \lim_{k\to\infty} \sum_{j=1}^{k} (\Delta_n(jt/k) -  \Delta_n((j-1)t/k))(D(jt/k)).
	\end{align*}
\end{lem}
\begin{proof}
	The proof is completely analogous to \cite[Lemma 4.2]{bh1}, i.e., it follows by a direct approximation with piecewise integrals, where $D(s)$ is kept constant. 
\end{proof}

\begin{prop}
	\label{th:main prep}
	For all $t \in \IR$ and $f \in \hs$, we have $\alpha_t(a^*(f)) \in \Afe_1$ and  $t \mapsto \alpha_t(a^*(f))$ is continuous in the seminorms topology. 
\end{prop}
\begin{proof}
	Again, we use the Dyson series \eqref{eq:dyson}
	\begin{align}
		\label{eq:dyson2}
		\gamma_t( a^*(f))|_{\FS_n} = \sum_{l=0}^\infty D_{n,l}(t),
	\end{align}
	where $D_{n,0}(t) := a^*(f)|_{\FS_n}$, and
	\begin{align}
		\label{eq:Dl recursion formula}
		D_{n,l+1}(t) := - \i \int_0^t [D_ {n,l}(s),\Vb(s)]|_{\FS_n} \d s.
	\end{align}
	The recursion formula \eqref{eq:Dl recursion formula} yields
	\begin{align}
		\label{eq:Dl estimate}
		\nn{ D_{n,l}(t) } \leq  \frac{\abs t^l}{l!}( \nn{V_{n}}+ \nn{V_{n+1}})^l \nn f,
	\end{align}
	which shows that the Dyson series \eqref{eq:dyson2} converges absolutely and  that $t \mapsto D_{n,l}(t)$ is continuous. With \cref{th:approx deltas Daf integral}, we obtain
	\begin{align}
		\label{eq:Dnl C formula}
		D_{n,l}(t) = &\lim_{k_1 \to\infty} \ldots \lim_{k_l \to \infty} C^{k_1, \ldots, k_{l}}_{n,l}(t),
	\end{align}
	where $C_{n,0}(t) = a^*(f)|_{\FS_n}$ and
	 \[C^{k_1, \ldots, k_{l+1}}_{n,l+1}(t) = \sum_{j=1}^{k_{l+1}} (\Delta_n(jt/k_{l+1}) -  \Delta_n((j-1)t/k_{l+1}))(C^{k_1, \ldots, k_{l}}_{n,l}(t)).\] 
	 We now show by induction that for all $l \in \IN_0$, $k_1, \ldots, k_l \in \IN$, and all $n \in \IN$,
	\begin{align}
		\label{eq:Cl statement}
		(C^{k_1, \ldots, k_{l}}_{m,l})_{m=0}^n \in \Af_1|_{\FS_{\leq n}} = \lin\{ a^*(f) B : f \in \hs,~B \in \Af_0 \} |_{\FS_{\leq n}}.
	\end{align}
	For $l=0$ this is clear. The induction step follows from
	\begin{align*}
		 &\left[ a^*(f) B , \int_0^t \Vb(s) \d s \right]\bigg|_{\FS_n}  =  a^*(f) \left[B,  \int_0^t V_n(s) \d s\right]|_{\FS_n}   \\ &\quad + \frac{1}{\nn f^2}  a^*(f) a(f) \left[a^*(f),  \int_0^t \Vb(s) \d s\right] B|_{\FS_n}  + \frac{1}{\nn f^2} \left[ a^*(f) a(f),  \int_0^t V_{n+1}(s) \d s \right]  a^*(f) B|_{\FS_n}.
	\end{align*} 
	By \cref{th:kappa Delta compatibility} we know that the first and the third term, and by \cref{th:Deltat af restriction} also the second term, are products of $a^*(f)$ and elements in $\Af_0$. Therefore, all of them lie in $\Af_1|_{\FS_{\leq n}}$.
	
	From \eqref{eq:Dnl C formula} and \eqref{eq:Cl statement} it follows that $(D_{m,l}(t))|_{m=0}^n \in \Afe_1|_{\FS_{\leq n}}$ for all $n$ and $l$. Thus, we infer that $\gamma_t(a^*(f))|_{\FS_{\leq n}} \in \Afe|_{\FS_{\leq n}}$ for all $n$, and therefore, $\gamma_t(a^*(f)) \in \Afe$ and hence also $\alpha_t(a^*(f)) \in \Afe$.  The continuity in $t$ follows from the continuity of the $D_{n,l}(t)$. 
\end{proof}

\begin{proof}[Proof of \cref{th:main}]
	From \cref{th:oplus algebra}, \cref{th:main prep} and the definition \eqref{eq:Afe characterization} we know that, for any $k \in \IZ$ and $A \in \Af_k$, we have $\alpha_t(A) \in \Afe_k$, $t \in \IR$, and the map $t \mapsto \alpha_t(A)$ is continuous in the $\nn \cdot _n$ seminorms. The same is true if $A \in \Afe_k$, since then we can find a sequence $(A_l)_{l \in \IN}$ in $\Af_k$ such that $A_l \to A$ in the seminorms, which implies $\alpha_t(A_l) \to \alpha_t(A)$ in the seminorms. 
	
	  If $A \in \Afe$ is arbitrary, there is a sequence $(A_l)_{l \in \IN}$ in $\bigoplus_{k \in \IZ} \Afe_k$ such that $A_l \to A$ in operator norm. Since the previous discussion showed that $\alpha_t(A_l) \in \bigoplus_{k \in \IZ} \Afe_k$ for all $l$, we have $\alpha_t(A) = \lim_{l\to\infty} \alpha_t(A_l) \in\Afe$ in operator norm as well. With an $\epsilon/3$-argument it also follows that $t \mapsto \alpha_t(A)$ is continuous in the seminorms.  
\end{proof}
\begin{proof}[Proof of \cref{th:main coro}]
	For all $A \in \Afe_k$, $k \in \IZ$, and  $f \in \Cci(\IR)$, we again consider the integrals $\int f(s) \alpha_s(A) \d s \in \Lb(\FS)$. With the same arguments as in the proof of this corollary for the particle number preserving case, we find that these integrals lie in $\Afe_k$. Let $\Afe_H$ be the $C^*$-subalgebra of $\Afe$ generated by $\int f(s) \alpha_s(A) \d s$, $f \in \Cci(\IR)$, $A \in \Afe_k$, $k \in \IZ$. Strong continuity on $\Afe_H$ follows with the same estimate \eqref{eq:pettis estimate} as before, and \eqref{eq:delta_k estimate} also holds for all $A \in \Afe_k$  and $k \in \IZ$. Therefore, $\Afe_H$ is dense in $\bigoplus_{k \in \IZ}\Afe_k$ and thus, also in $\Afe$. 
\end{proof}

\subsection*{Acknowledgement}

I would like to thank Detlev Buchholz, Marius Lemm and Melchior Wirth for some comments and discussions. 

\subsection*{Funding}

Partial financial support was received from Deutsche Forschungsgemeinschaft (DFG, German Research Foundation) - TRR 352 - Project-ID 470903074.

\subsection*{Data availability}
Data sharing not applicable to this article as no datasets were generated or analyzed during the current study.

\section*{Declarations}

\subsection*{Competing interests}
The author has no relevant financial or non-financial interests to disclose.

\printbibliography

@article{bratteli1972inductive,
	title={Inductive limits of finite dimensional C*-algebras},
	author={Bratteli, Ola},
	journal={Trans. Am. Math. Soc.},
	volume={171},
	pages={195--234},
	year={1972},
	doi={10.2307/1996380}
}

@article{haag1967equilibrium,
	title={On the equilibrium states in quantum statistical mechanics},
	author={Haag, Rudolf and Hugenholtz, Nico M and Winnink, Marinus},
	journal={Commun. Math. Phys.},
	volume={5},
	number={3},
	pages={215--236},
	year={1967},
	publisher={Springer},
	doi={10.1007/BF01646342}
}

@article{araki_wyss,
	title={Representations of canonical anticommutation relations},
	author={Araki, Huzihiro and Wyss, Walter},
	journal={Helv. Phys. Acta},
	volume={37},
	number={2},
	pages={136--159},
	year={1964},
	doi={10.5169/seals-113476}
}

@book{Arai.2018,
    address = {New Jersey},
    author = {Asao Arai},
    doi = {10.1142/10367},
    publisher = {World Scientific},
    subtitle = {An Introduction to Mathematical Analysis of Quantum Fields},
    title = {Analysis on {F}ock Spaces and Mathematical Theory of Quantum Fields},
    year = {2018}
}

@book{attal2006open,
    author = {Attal, St{\'e}phane},
    doi = {10.1007/b128449},
    publisher = {Springer Science \& Business Media},
    title = {Open Quantum Systems I: The Hamiltonian Approach},
    volume = {1},
    year = {2006}
}

@article{bh,
    author = {Buchholz, Detlev and Grundling, Hendrik},
    doi = {10.1016/j.jfa.2008.02.011},
    journal = {J. Funct. Anal.},
    number = {11},
    pages = {2725--2779},
    publisher = {Elsevier},
    title = {The resolvent algebra: A new approach to canonical quantum systems},
    volume = {254},
    year = {2008}
}

@article{bh1,
    author = {Buchholz, Detlev},
    doi = {10.1007/s00220-018-3144-6},
    journal = {Commun. Math. Phys.},
    number = {3},
    pages = {949--981},
    publisher = {Springer},
    title = {The resolvent algebra of non-relativistic Bose fields: observables, dynamics and states},
    volume = {362},
    year = {2018}
}

@article{bh2,
    author = {Buchholz, Detlev},
    doi = {10.1007/s00220-019-03629-8},
    journal = {Commun. Math. Phys.},
    number = {2},
    pages = {1159--1199},
    publisher = {Springer},
    title = {The resolvent algebra of non-relativistic Bose fields: sectors, morphisms, fields and dynamics},
    volume = {375},
    year = {2020}
}

@book{BR1,
    address = {Berlin},
    author = {Ola Bratteli and Derek William Robinson},
    doi = {10.1007/978-3-662-02520-8},
    edition = {2nd},
    publisher = {Springer},
    series = {Texts and Monographs in Physics},
    title = {Operator Algebras and Quantum Statistical Mechanics 1: C*- and W*-Algebras Symmetry Groups Decomposition of States},
    year = {1987}
}

@book{BR2,
    address = {Berlin},
    author = {Ola Bratteli and Derek William Robinson},
    doi = {10.1007/978-3-662-09089-3},
    edition = {2nd},
    publisher = {Springer},
    series = {Texts and Monographs in Physics},
    title = {Operator Algebras and Quantum Statistical Mechanics 2: Equilibrium States. Models in Quantum Statistical Mechanics},
    year = {1997}
}

@article{garcia2023symmetric,
    author = {Garcia, Stephan Ramon and O’Loughlin, Ryan and Yu, Jiahui},
    doi = {10.4153/S0008414X23000901},
    journal = {Can. J. Math.},
    pages = {1--23},
    publisher = {Canadian Mathematical Society},
    title = {Symmetric and antisymmetric tensor products for the function-theoretic operator theorist},
    year = {2023}
}

@article{GebertNachtergaeleReschkeSims.2020,
    author = {Gebert, Martin and Nachtergaele, Bruno and Reschke, Jake and Sims, Robert},
    doi = {10.1007/s00023-020-00959-5},
    eprint = {1912.12552},
    fjournal = {Annales Henri Poincar\'{e}},
    journal = {Ann. Henri Poincar\'{e}},
    number = {11},
    pages = {3609-3637},
    title = {{L}ieb--{R}obinson Bounds and Strongly Continuous Dynamics for a Class of Many-Body Fermion Systems in {$\ifx\mathbb\undefined{R}\else{\mathbb R}^d\fi$}},
    volume = {21},
    year = {2020}
}

@article{gluza2016equilibration,
    author = {Gluza, Marek and Krumnow, Christian and Friesdorf, Mathis and Gogolin, Christian and Eisert, Jens},
    doi = {10.1103/PhysRevLett.117.190602},
    journal = {Phys. Rev. Lett.},
    number = {19},
    pages = {190602},
    publisher = {APS},
    title = {{Equilibration via Gaussification in fermionic lattice systems}},
    volume = {117},
    year = {2016}
}

@Article{hinrichs2024lieb,
	author={Hinrichs, Benjamin
	and Lemm, Marius
	and Siebert, Oliver},
	title={On Lieb--Robinson Bounds for a Class of Continuum Fermions},
	journal={Ann. Henri Poincar{\'e}},
	year={2025},
	month={Jan},
	day={01},
	volume={26},
	number={1},
	pages={41-80},
	doi={10.1007/s00023-024-01453-y}
}

@article{irreversible,
    author = {Nachtergaele, Bruno and Vershynina, Anna and Zagrebnov, Valentin A},
    doi = {10.1090/conm/552/10916},
    journal = {AMS Contemporary Mathematics},
    pages = {161--175},
    title = {Lieb-Robinson bounds and existence of the thermodynamic limit for a class of irreversible quantum dynamics},
    volume = {552},
    year = {2011}
}

@article{jakel1995uniqueness,
    author = {J{\"a}kel, Christian D},
    journal = {Lett. Math. Phys.},
    pages = {87--97},
    publisher = {Springer},
    title = {On the uniqueness of the equilibrium state for an interacting fermion gas at high temperatures and low densities},
    volume = {33},
    year = {1995},
    DOI = {10.1007/BF00750814}
}

@article{jakpillet3,
    author = {Jakšić, V. and Pillet, C.-A.},
    journal = {Comm. Math. Phys.},
    number = {3},
    pages = {627--651},
    publisher = {Springer},
    title = {On a model for quantum friction. {III. Ergodic} properties of the spin-boson system},
    volume = {178},
    year = {1996},
    doi = {10.1007/BF02108818}
}

@article{jakvsic2002non,
    author = {Jak{\v{s}}i{\'c}, Vojkan and Pillet, C-A},
    journal = {Commun. Math. Phys.},
    pages = {131--162},
    publisher = {Springer},
    title = {Non-Equilibrium Steady States of Finite Quantum Systems Coupled to Thermal Reservoirs},
    volume = {226},
    year = {2002},
    DOI = {10.1007/s002200200602}
}

@article{LiebRobinson.1972,
    author = {Lieb, Elliott H. and Robinson, Derek William},
    doi = {10.1007/BF01645779},
    fjournal = {Communications in Mathematical Physics},
    journal = {Commun. Math. Phys.},
    pages = {251-257},
    title = {The Finite Group Velocity of Quantum Spin Systems},
    volume = {28},
    year = {1972}
}

@incollection{NachtergaeleSimsYoung.2018,
    author = {Nachtergaele, Bruno and Sims, Robert and Young, Amanda},
    booktitle = {Mathematical Problems in Quantum Physics},
    doi = {10.1090/conm/717/14443},
    eprint = {1705.08553},
    pages = {93-115},
    publisher = {AMS},
    series = {Contemporary Mathematics},
    title = {Lieb-{R}obinson bounds, the spectral flow, and stability of the spectral gap for lattice fermion systems},
    volume = {717},
    year = {2018}
}

@article{narnhofer1990quantum,
    author = {Narnhofer, H and Thirring, W},
    doi = {10.1103/PhysRevLett.64.1863},
    journal = {Phys. Rev. Lett.},
    number = {16},
    pages = {1863},
    publisher = {APS},
    title = {Quantum field theories with Galilei-invariant interactions},
    volume = {64},
    year = {1990}
}

@article{narnhofer1991galilei,
    author = {Narnhofer, Heide and Thirring, Walter},
    doi = {10.1142/S0217751X91001453},
    journal = {Int. J. Mod. Phys. A},
    number = {17},
    pages = {2937--2970},
    publisher = {World Scientific},
    title = {Galilei-invariant quantum field theories with pair interactions: A review},
    volume = {6},
    year = {1991}
}

@article{narnhofer2020local,
    author = {Narnhofer, Heide},
    title = {Local Normality for KMS states with Galilei invariant Interaction},
    year = {2020},
    journal = {arXiv:2011.01059},
    doi = {10.48550/arXiv.2011.01059}
}

@book{ReedSimon.1979,
    address = {San Diego},
    author = {Michael Reed and Barry Simon},
    doi = {10.1007/BFb0079274},
    publisher = {Academic Press},
    series = {Methods of Modern Mathematical Physics},
    title = {Scattering Theory},
    volume = {3},
    year = {1979}
}

@article{td1,
    author = {Nachtergaele, Bruno and Ogata, Yoshiko and Sims, Robert},
    doi = {10.1007/s10955-006-9143-6},
    journal = {J. Stat. Phys.},
    pages = {1--13},
    publisher = {Springer},
    title = {Propagation of correlations in quantum lattice systems},
    volume = {124},
    year = {2006}
}

@article {td2,
	AUTHOR = {Amour, L. and L\'evy-Bruhl, P. and Nourrigat, J.},
	TITLE = {Dynamics and {L}ieb-{R}obinson estimates for lattices of
	interacting anharmonic oscillators},
	JOURNAL = {Colloq. Math.},
	FJOURNAL = {Colloquium Mathematicum},
	VOLUME = {118},
	YEAR = {2010},
	NUMBER = {2},
	PAGES = {609--648},
	DOI = {10.4064/cm118-2-17}
}

@article{td3,
    author = {Nachtergaele, Bruno and Schlein, Benjamin and Sims, Robert and Starr, Shannon and Zagrebnov, Valentin},
    doi = {10.1142/S0129055X1000393X},
    journal = {Rev. Math. Phys.},
    number = {02},
    pages = {207--231},
    publisher = {World Scientific},
    title = {On the existence of the dynamics for anharmonic quantum oscillator systems},
    volume = {22},
    year = {2010}
}

@article{td4,
    author = {Bachmann, Sven and Michalakis, Spyridon and Nachtergaele, Bruno and Sims, Robert},
    doi = {10.1007/s00220-011-1380-0},
    journal = {Commun. Math. Phys.},
    number = {3},
    pages = {835--871},
    publisher = {Springer},
    title = {Automorphic equivalence within gapped phases of quantum lattice systems},
    volume = {309},
    year = {2012}
}

@article{thermalionization,
    author = {Fr{\"o}hlich, J.
and Merkli, M.},
    day = {01},
    doi = {10.1023/B:MPAG.0000034613.13746.8a},
    issn = {1572-9656},
    journal = {Math. Phys. Anal. Geom.},
    month = {Aug},
    number = {3},
    pages = {239--287},
    title = {Thermal Ionization},
    volume = {7},
    year = {2004}
}

@book{sakai1991operator,
	title={Operator algebras in dynamical systems},
	author={Sakai, Sh{\=o}ichir{\=o}},
	year={1991},
	doi={10.1017/cbo9780511662218},
	publisher={Cambridge University Press}
}

@Article{bachmann2024liebrobinson,
	author={Bachmann, Sven
	and Nittis, Giuseppe De},
	title={Lieb--Robinson Bounds in the Continuum Via Localized Frames},
	journal={Ann. Henri Poincar{\'e}},
	year={2025},
	month={Jan},
	day={01},
	volume={26},
	number={1},
	pages={1-40},
	doi={10.1007/s00023-024-01511-5}
}

@article{buchholzyngvason2024,
	author = {Buchholz, Detlev and Yngvason, Jakob},
	title = {Many-body physics and resolvent algebras},
	year = {2024},
	journal = {arXiv: 2411.04737},
	doi = {10.48550/arXiv.2411.04737}
}

@book{treves,
	title={Topological Vector Spaces, Distributions and Kernels},
	author={Tr{\`e}ves, F.},
	lccn={66030175},
	series={Pure and Applied Mathematics - Academic Press},
	year={1970},
	publisher={Academic Press}
}

\end{document}